\documentclass[manuscript]{acmart}

\setcopyright{none}
\renewcommand\footnotetextcopyrightpermission[1]{}

\usepackage{bm}
\usepackage{color}
\usepackage{graphicx}
\usepackage{booktabs}
\usepackage{multirow}
\usepackage{threeparttable}
\usepackage{enumerate}
\usepackage[english]{babel}
\usepackage{algorithmicx}
\usepackage{algorithm}
\usepackage{algpseudocode}
\usepackage{nomencl}
\usepackage{multicol}
\usepackage{commath}
\usepackage{makecell}
\usepackage{blindtext}
\usepackage{soul}
\usepackage{comment}
\usepackage{braket}
\newtheorem{theorem}{Theorem}[section]
\newtheorem{lemma}[theorem]{Lemma}
\newtheorem{corollary}[theorem]{Corollary}
\theoremstyle{remark}
\newtheorem{remark}{Remark}[section]

\newcommand{\cutval}{\ensuremath{c_{\mathrm{cut}}}}

\graphicspath{{Figures/}{ieee_bus_systems_no_legends_large_font_outputs/pngs/}}
\begin{document}
\title{REGRID-QAOA: A Resource-Efficient Hybrid QAOA Framework for Physics-Constrained Power System Islanding}

\author{Yuqi Jiang}
\affiliation{%
  \institution{Pennsylvania State University}
  \department{Department of Electrical Engineering}
  \city{University Park}
  \state{Pennsylvania}
  \country{USA}}

\author{Yuqi Zhang}
\affiliation{%
  \institution{Kent State University}
  \department{Department of Computer Science}
  \city{Kent}
  \state{Ohio}
  \country{USA}}

\author{Zhiding Liang}
\affiliation{%
  \institution{Rensselaer Polytechnic Institute}
  \department{Department of Computer Science}
  \city{Troy}
  \state{New York}
  \country{USA}}

\author{Qiang Guan}
\affiliation{%
  \institution{Kent State University}
  \department{Department of Computer Science}
  \city{Kent}
  \state{Ohio}
  \country{USA}}

\author{Yan Li}
\email{yql5925@psu.edu}
\affiliation{%
  \institution{Pennsylvania State University}
  \department{Department of Electrical Engineering}
  \city{University Park}
  \state{Pennsylvania}
  \country{USA}}

\author{Ganesh Kumar Venayagamoorthy}
\affiliation{%
  \institution{Clemson University}
  \department{Real-Time Power and Intelligent Systems Laboratory, Holcombe Department of Electrical and Computer Engineering}
  \city{Clemson}
  \state{South Carolina}
  \country{USA}}
\affiliation{%
  \institution{University of Pretoria}
  \city{Pretoria}
  \country{South Africa}}

\begin{abstract}
Quantum computing has rapidly emerged as a powerful paradigm for tackling computationally demanding problems. In particular, quantum optimization shows strong promise for hard combinatorial problems in power systems, where increasing distributed energy penetration heightens the need for intentional islanding to maintain grid reliability and resilience. However, power system islanding is an NP-hard combinatorial optimization problem that becomes computationally prohibitive for classical solvers as network size grows, motivating the use of quantum computing as a promising alternative pipeline. This study develops a resource-efficient hybrid QAOA islanding framework that brings physics-constrained power-system partitioning into the quantum optimization workflow. The framework combines coherency-informed graph reduction, physics-aware constraint modeling, and structured post-processing to efficiently convert shallow-circuit QAOA samples into high-quality feasible islanding decisions without deep circuits or large shot budgets. The proposed framework is validated on the standard IEEE benchmark systems (9-, 14-, 24-, 30-, 39-, and 57-bus), demonstrating that the hybrid workflow achieves Gurobi-optimal solution quality with a clear quantum resource advantage over vanilla QAOA, while the resulting islanding solutions satisfy all physical feasibility requirements after network separation. This study establishes QAOA-based islanding as a viable quantum approach for critical infrastructure, with structured post-processing as the key enabler of quantum resource efficiency.
\end{abstract}

\ccsdesc[500]{Hardware~Quantum computation}
\ccsdesc[500]{Mathematics of computing~Combinatorial optimization}
\ccsdesc[300]{Applied computing~Engineering}

\keywords{Power system islanding, Quantum Approximate Optimization Algorithm (QAOA), post-processing}

\maketitle
\section{Introduction}
Recent advances in quantum hardware and algorithms have intensified interest in applying near-term quantum computing to structured optimization and simulation problems \cite{montanaro2016quantum,preskill2018quantum}. In the noisy intermediate-scale quantum (NISQ) setting, hybrid quantum-classical methods are especially attractive because they combine parameterized quantum circuits with classical outer-loop optimization while avoiding the requirement of full fault tolerance \cite{bharti2022nisq,cerezo2021variational}. Among these approaches, variational quantum algorithms have emerged as a leading framework for exploring practical quantum advantage in scientifically and technologically relevant applications \cite{cerezo2021variational}.

Among the engineering domains that can benefit from such hybrid quantum optimization, power systems are especially compelling because many power grid operation decisions are inherently large-scale combinatorial problems defined on network topologies. Tasks such as topology reconfiguration \cite{aziz2021review}, contingency mitigation \cite{che2018preventive}, and controlled islanding \cite{kyriacou2017controlled} require operators to search over many discrete operating configurations while simultaneously satisfying physical and operational constraints. These topology-constrained decision problems are often NP-hard and can become computationally prohibitive for classical methods as network size and operational uncertainty increase, particularly when high-quality decisions are needed under stressed operating conditions \cite{auyeung2023nphard}. This computational structure makes modern power grids a natural application area for quantum optimization methods \cite{morstyn2024opportunities}.

Large-scale integration of distributed energy resources has made resilient power system operation more challenging by introducing new sources of uncertainty \cite{dincer2000renewable}. The intermittent output of renewable energy resources such as wind and solar can cause significant fluctuations, making it challenging to maintain the balance between supply and demand. If not properly managed, these imbalances can propagate through the interconnected grid and escalate into cascading failures \cite{guo2017critical}. As a result, protection schemes are often required to adopt more conservative settings to preserve security margins. This motivates the necessity of intentional islanding as a system-level resilience measure: by deliberately separating a stressed power network into multiple self-sustained islands, operators can confine disturbances, preserve critical loads, and create sustainable operating regions while restoration is initiated \cite{ding2015mixed}.

However, practical islanding is fundamentally a combinatorial decision-making problem under tight operational constraints \cite{ding2015mixed}. This combinatorial optimization problem is NP-hard, which requires substantial classical computing resources and implementation time to obtain near-optimal results. In power system studies, classical optimization methods is the main idea for solving this problem. 
Prior work has approached controlled islanding using Linear Programming \cite{kyriacou2017controlled}, Ordered Binary Decision Diagrams \cite{sun2003splitting}, and Binary Integer Programming \cite{huang2013optimal}. These classical approaches can become computationally prohibitive and less adaptable as system size and uncertainty grow, limiting their ability to deliver high-quality islanding decisions.
Although quantum computing has been applied to power grid partitioning problems related to islanding \cite{hartmann2025partitioning}, these approaches model the problem as a graph-theoretic decomposition and do not enforce the operational constraints that practical islanding requires. The quantum annealing formulation in \cite{hartmann2025partitioning} optimizes only topological objectives such as cut cost and sub-network size balance, while omitting power balance and the connectivity of the resulting sub-networks, leaving solutions potentially infeasible as operational islands.

The recent progress of quantum computing has provided a new pipeline to solve the combinatorial optimization problems \cite{chicano2024combinatorial}. It leverages quantum superposition and interference to lead the search toward high-quality solutions, potentially altering the effective optimization landscape relative to classical procedures \cite{gemeinhardt2024nisq}. The Quantum Approximate Optimization Algorithm (QAOA) has been studied in a wide range of optimization problems \cite{farhi2014qaoa}. Recent studies have applied QAOA to financial portfolio optimization \cite{zaman2024po}, chemical simulation and drug discovery \cite{mustafa2022variational}, and traffic flow optimization \cite{villanueva2025hybrid}. In the context of power system applications, 
QAOA has also been used for the OPP problem \cite{10980373}, unit commitment benchmarking \cite{adler2025scaling}, and power system contingency analysis \cite{11250180}. 

A central challenge in applying QAOA to constrained problems is that quantum measurement produces raw bitstring samples, which are not guaranteed to satisfy the problem's feasibility requirements \cite{niroula2022constrained}. Post-processing addresses this gap by applying classical procedures after quantum measurement to convert these raw samples into feasible, high-quality solutions, for instance, by repairing constraint violations, selecting the best sample among the measured outcomes, or applying local search corrections \cite{blekos2024review,shirai2024postprocessing}. Despite being a necessary step in any constrained quantum optimization pipeline, post-processing remains only sparsely studied. Prior work has shown that quadratic unconstrained binary optimization (QUBO)-based greedy post-processing can map quantum-device outputs to constraint-satisfying solutions while preserving monotone descent of a selected energy \cite{shirai2024postprocessing}. That framework, however, assumes that constraints are independent in the sense that each decision variable appears in at most one constraint penalty. While this holds for simple combinatorial benchmarks, many real-world optimization applications, including power system islanding, involve operationally coupled constraints in which the same variable simultaneously governs multiple requirements, violating the independence assumption and placing such problems outside the scope of existing post-processing theory. In addition, the previous QUBO-based post-processing method does not cover feasibility requirements that must be evaluated outside the QUBO penalty model. Power system islanding post-processing is distinguished from both of the preceding settings in that its feasibility conditions are simultaneously operationally coupled and only partially representable within the QUBO penalty model.

In this study, an end-to-end QAOA framework is developed for the power system islanding problem, which aims to minimize the mismatch of the power flow during islanding operations. This paper pioneers the study of QAOA for the full power system islanding solution pipeline. Beyond the quantum modeling and optimization stages, the paper also develops post-processing as a central part of the solution framework, proposed for settings where part of the feasible structure lies outside standard QUBO encodings and where constraints are operationally coupled, extending principled post-processing beyond the independent-constraint settings addressed by prior work. The paper therefore also studies a broader methodological question for quantum computing in real-world optimization: how sampled quantum outputs can be systematically converted into feasible, high-quality decisions when the governing constraints combine QUBO-encodable penalties with application-level structure that must be checked or repaired classically. The proposed framework is validated on six IEEE benchmark systems, 9-, 14-, 24-, 30-, 39-, and 57-bus using real IBM Quantum processors \texttt{ibm\_marrakesh}, demonstrating that principled post-processing consistently recovers near-optimal, fully feasible islanding partitions across all tested instances with modest circuit resources. A complementary noise study further indicates that the workflow remains resilient under realistic quantum noise, supporting its practicality for near-term quantum hardware. The main contributions of this paper can be summarized as:
\begin{itemize}
    \item This study develops a resource-efficient hybrid QAOA framework for physics-constrained power system islanding, advancing quantum islanding beyond purely topological graph partitioning. The formulation jointly represents the objective of minimizing disrupted active power flow across severed lines and the operational feasibility requirements needed for valid islands. A coherency-based graph reduction further lowers qubit count before circuit execution without compromising solution quality.
    \item This study provides a comprehensive post-processing study for constrained quantum optimization applied to islanding, systematically designing and comparing multiple repair strategies that convert raw QAOA samples into fully feasible islanding decisions under both QUBO-encodable and operationally coupled constraints. The study specifically addresses the challenge of enforcing exact graph connectivity, which lies outside standard QUBO encodings and must be handled classically, alongside the algebraic operational constraints. Unlike existing post-processing approaches, the proposed methods and their guarantees remain applicable when constraints are operationally coupled and the independence assumption of prior theory does not hold, directly broadening the reach of principled post-processing to practical applications beyond simple benchmarks.
    \item The framework is validated on six IEEE benchmark systems (9- to 57-bus) using real IBM Quantum processors. The coherency-based reduction demonstrably reduces the required qubit count across the benchmark cases without altering solution optimality. The post-processing stage substantially improves both solution quality and feasibility over vanilla QAOA sampling, with the best method recovering the Gurobi-optimal cut on all benchmarks using few QAOA circuit layers and modest shot budgets while producing fully feasible islanding solutions. Additional noise-resilience evaluation further shows that the workflow maintains robust performance under realistic quantum noise.
\end{itemize}

The remainder of this paper is as follows: Section \ref{formulation} introduces the islanding problem formulation, Section \ref{quantum} develops the quantum solution process for the islanding formulation, Section \ref{sec:postprocessing} presents the post-processing framework for non-QUBO and feasibility-critical constraints, Section \ref{sec:numerical} gives a numerical example, and Section \ref{conclusion} summarizes this work. 

\section{Problem Formulation}
\label{formulation}

A power grid is modeled as a weighted graph $\mathcal{G} = (V, E)$, where $V = \{1, \dots, N\}$ is the set of buses and $E \subseteq V \times V$ is the set of transmission lines. Each transmission line $e_{ij} \in E$ connecting bus $i$ to bus $j$ carries an associated weight $p_{ij}$, representing the active power flow along that line in the direction from $i$ to $j$. The graph structure captures the physical topology of the network, where buses represent substations or connection points, and transmission lines represent the physical conductors through which power is transferred.

To distinguish bus functionalities, the bus set $V$ is partitioned into three mutually disjoint subsets:
\begin{equation}
    V = V_G \cup V_L \cup V_0,
\end{equation}
where $V_G$, $V_L$, and $V_0$ denote the sets of generator, load, and transit buses, respectively. Generator buses inject active power into the network, load buses withdraw active power to serve demand, and transit buses act purely as intermediate routing points, facilitating power transfer without contributing to or consuming from the network.

For each bus $i \in V$, let $p_i^{G} \geq 0$ and $p_i^{D} \geq 0$ denote the total active power generation and demand, respectively. The net active power injection at bus $i$ is then defined as
\begin{equation}
    p_i := p_i^{G} - p_i^{D},
\end{equation}
where $p_i > 0$ indicates a net supplier (generator bus), $p_i < 0$ indicates a net consumer (load bus), and $p_i = 0$ indicates a transit bus. This work neglects network losses, so system-wide power balance requires the following conservation constraint:
\begin{equation}
    \sum_{i \in V} p_i = \sum_{i \in V} p_i^{G} - \sum_{i \in V} p_i^{D} = 0.
\end{equation}
Accordingly, the three bus subsets are formally defined as
\begin{align}
    V_G &:= \{i \in V : p_i^{G} > 0,\ p_i^{D} = 0\}, \\
    V_L &:= \{i \in V : p_i^{D} > 0,\ p_i^{G} = 0\}, \\
    V_0 &:= \{i \in V : p_i^{G} = 0,\ p_i^{D} = 0\}.
\end{align}
This partition provides a structured characterization of the functional roles of buses within the grid.

When a disturbance occurs and islanding operation becomes necessary, the number of islands $K$ is pre-determined by the coherent groups of generators \cite{kyriacou2017controlled, chow1982time}. In this work, coherent generator groups are identified via a coherency algorithm \cite{chow1982time, you2004slow}. Accordingly, the generator-bus set $V_G$ is partitioned into $K$ pairwise disjoint coherent groups:
\begin{equation}
    V_G = \bigcup_{k=1}^{K} V_G^{(k)}, \qquad V_G^{(k)} \cap V_G^{(\ell)} = \varnothing, \quad \forall\, k \neq \ell,
\end{equation}
where $V_G^{(k)}$ denotes the $k$-th coherent generator group. By construction, all generators within the same group $V_G^{(k)}$ are enforced to remain within the same island following controlled islanding.

Each island is represented as a connected subgraph $\mathcal{G}^{(k)} = (V^{(k)}, E^{(k)})$ of $\mathcal{G}$, where $V^{(k)}$ denotes the full set of buses assigned to island $k$, comprising generator, load, and transit buses. The island bus sets form a complete partition of $V$:
\begin{equation}
    V = \bigcup_{k=1}^{K} V^{(k)}, \qquad V^{(k)} \cap V^{(\ell)} = \varnothing, \quad \forall\, k \neq \ell,
\end{equation}
with the coherency constraint $V_G^{(k)} \subseteq V^{(k)}$ ensuring that generators within the same coherent group are assigned to the same island. The set of transmission lines severed to form the islands, referred to as the \textit{cut set}, is defined as:
\begin{equation}
    E_c := \{e_{ij} \in E : i \in V^{(k)},\ j \in V^{(\ell)},\ k \neq \ell\},
\end{equation}
here $E_c$ comprises all transmission lines whose two endpoints are assigned to distinct islands. To ensure each island operates self-sufficiently following separation, the net active power 
injection must be balanced within each island:
\begin{equation}
    \sum_{i \in V^{(k)}} p_i = 0, \quad \forall\, k = 1, \dots, K.
\end{equation}
The controlled islanding problem is then formulated as minimizing the total disrupted 
active power flow across the severed lines, subject to the partitioning, connectivity, 
and coherency requirements, which can be expressed as:
\allowdisplaybreaks
\begin{align}
\min_{\{V_k\}_{k=1}^{K}}
\;&\sum_{\substack{k,h\in\{1,\dots,K\}\\ h\neq k}}
\;\sum_{\substack{i\in V_k,\; j\in V_h\\ e_{ij}\in E}}
\frac{|p_{ij}|+|p_{ji}|}{2}, \label{eq:islanding_objective}\\
\text{s.t. }\;
& V_k\cap V_h=\emptyset,\ \forall k\neq h,\qquad \bigcup_{k=1}^{K} V_k = V, \label{con:onehot}\\
& |V_k| \ge N_{\min},\ \forall k\in\{1,\dots,K\}, \label{con:min_size}\\
& |V_k\cap V_G| \ge N_{G,\min},\ \forall k\in\{1,\dots,K\}, \label{con:gen}\\
& |V_k\cap V_L| \ge N_{L,\min},\ \forall k\in\{1,\dots,K\}, \label{con:load}\\
& \forall g,g'\in V_G:\ \Big( c(g)=c(g')\Rightarrow \exists k:\{g,g'\}\subseteq V_k \Big)\ \wedge \notag\\
& \hspace{10mm}\Big( c(g)\neq c(g')\Rightarrow \forall k:\ \neg(\{g,g'\}\subseteq V_k) \Big), \label{con:coherency}\\
& \mathcal{C}(G[V_k]) = 1,\ \forall k\in\{1,\dots,K\}. \label{con:connect}
\end{align}
where the objective in~\eqref{eq:islanding_objective} minimizes the total magnitude of active power transfers across island boundaries. The parameters $N_{G,\min}\ge 1$ and $N_{L,\min}\ge 1$ denote the required minimum numbers of generator and load buses in each island, respectively.
To ensure a smooth post-islanding transition, the islanding decisions must satisfy key physical and operational constraints of the power system, which are given in ~(\ref{con:onehot})--(\ref{con:connect}). Here, $c(g)\in\{1,\dots,K\}$ denotes the coherent-group label of generator $g$. (\ref{con:onehot}) enforces a one-hot bus assignment so that $\{V_k\}_{k=1}^K$ forms a disjoint and complete partition of $V$. (\ref{con:min_size}) prevents degenerate islands, namely trivially small partitions that contain too few buses to represent meaningful or operationally useful islands, by requiring each island to contain at least $N_{\min}$ buses. (\ref{con:gen}) guarantees self-sustained generation capability by enforcing at least $N_{G,\min}$ generator buses in each island. (\ref{con:load}) avoids generator-only islands by requiring at least $N_{L,\min}$ load buses in each island. (\ref{con:coherency}) preserves dynamic stability by grouping coherent generators within the same island while separating non-coherent groups. (\ref{con:connect}) ensures operability by constraining each island subgraph $G[V_k]$ to be connected, and $\mathcal{C}$ is the graph connected check algorithm (the subgraph is connected if $\mathcal{C}(G[V_k]) = 1$). This formulation is computationally difficult because it contains weighted graph partitioning as its core combinatorial structure, and the additional coherency, size, generation, load, and connectivity constraints further restrict the feasible partitions. The feasible region is also highly nonconvex because the binary assignment, coherency, and connectivity conditions induce a discrete combinatorial search space with many local minima, making it difficult for classical solvers to reliably attain the global optimum within practical time limits.

While the formulation in~\eqref{eq:islanding_objective}--\eqref{con:connect} captures the essential combinatorial structure of the islanding problem, the cut objective is a proxy for physical feasibility rather than an exact measure: cutting lines with smaller pre-islanding active power flow generally causes less disturbance to the post-separation power balance. To assess the actual physical consistency of a partition beyond what this proxy captures, power flow validation is applied after the QAOA pipeline. As incorporating such constraints directly into the QUBO would significantly increase its complexity and qubit requirements, they are instead applied as two post-pipeline physical validation steps. First, a partition is accepted only if every island yields a feasible power flow solution. Second, each island must satisfy $P_{g,\max} > P_d$ and $Q_{\max} > Q_d$, where $P_{g,\max}$ and $Q_{\max}$ denote the total available active-generation capacity and reactive upper capability of the island, and $P_d$ and $Q_d$ denote the corresponding active- and reactive-power demand, ensuring each island is self-supporting in both active and reactive power following separation~\cite{li2010rrbalance}. The numerical results in Section~\ref{sec:numerical} confirm that this compromise does not affect solution feasibility.

The overall framework of the proposed approach is illustrated in Fig.~\ref{fig:overall_diagram}. Fig.~\ref{fig:overall_diagram}~(a) maps the real-world power grid into the graph representation used for islanding optimization; Fig.~\ref{fig:overall_diagram}~(b) and Fig.~\ref{fig:overall_diagram}~(c) show QUBO/Hamiltonian encoding followed by QAOA execution and measurement; and Fig.~\ref{fig:overall_diagram}~(d) and Fig.~\ref{fig:overall_diagram}~(e) show classical post-processing and extraction of the final islanding solution. The next section details the quantum solution process, including the coherency-based reduction applied before QUBO construction.

\begin{figure}[!t]
\centering
\includegraphics[width=\textwidth]{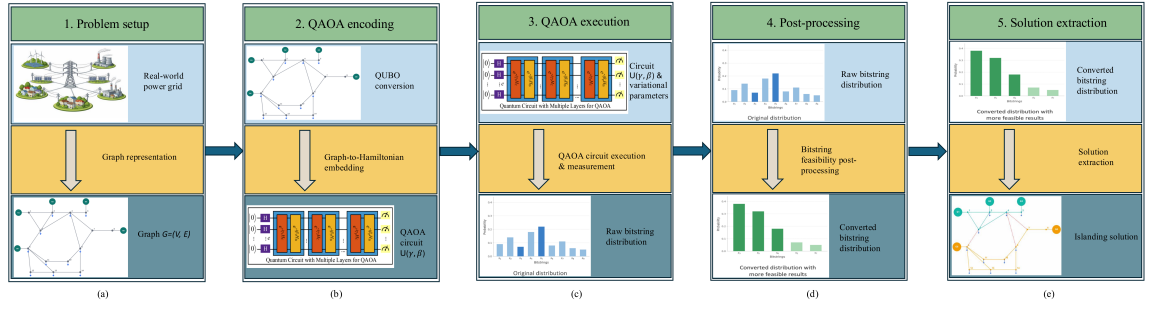}
\caption{Overall framework of the proposed QAOA-based controlled power system islanding approach. The pipeline consists of five operation blocks: (a)~problem setup, which converts real-world power-grid data into a graph representation; (b)~QAOA encoding, which maps the graph-partitioning problem into a Hamiltonian representation and circuit input; (c)~QAOA circuit execution and measurement, which produces a raw bitstring distribution; (d)~classical post-processing (M.~1--5), which repairs and filters sampled candidates to improve feasibility; and (e)~solution extraction, which converts selected bitstrings into the final islanding solution.}
\Description{Workflow diagram showing five stages: problem setup, QAOA encoding, QAOA execution, classical post-processing, and final solution extraction.}
\label{fig:overall_diagram}
\end{figure}

\section{QAOA Solutions for Power System Islanding}
\label{quantum}

Within the encoding and execution path summarized in Fig.~\ref{fig:overall_diagram}~(b)--(c), the quantum solution process consists of three stages: graph reduction via generator coherency, QUBO conversion of the objective and constraints, and QAOA circuit construction and optimization with a classical outer loop.

\subsection{Graph Reduction via Generator Coherency}
\label{sec:reduction}

In the NISQ setting, both the qubit count and the number of circuit executions (shots) needed to estimate the objective are limiting resources, and the sampling cost grows with the dimension of the search space. Embedding power system domain knowledge into the problem formulation offers a principled way to reduce this search space before the QUBO is constructed, saving quantum resources without sacrificing solution quality. The coherency constraint~(\ref{con:coherency}) provides exactly such an opportunity, since it requires every generator in $V_G^{(k)}$ to be assigned to the same island. When the subgraph $G[V_G^{(k)}]$ induced by group $k$ is connected, this constraint implies that the entire group can be collapsed into a single super-bus without altering the feasible set or the cut objective. This reduction is applied before the QUBO is constructed, decreasing both the number of assignment variables and the number of penalty terms in the Hamiltonian.

Let $\mathcal{K}_c \subseteq \{1, \ldots, K\}$ denote the indices of coherency groups whose induced subgraph is connected, verified by a breadth-first search on $G$ \cite{bundy1984breadth}. For each $k \in \mathcal{K}_c$, all buses in $V_G^{(k)}$ are replaced by a single super-bus $v^*_k$. The reduced bus set is
\begin{equation}
V' = \Bigl(V \setminus \bigcup_{k \in \mathcal{K}_c} V_G^{(k)}\Bigr)
     \cup \bigl\{v^*_k : k \in \mathcal{K}_c\bigr\},
\label{eq:reduced_V}
\end{equation}
with $N' = |V'| = N - \sum_{k \in \mathcal{K}_c}\bigl(|V_G^{(k)}| - 1\bigr) < N$. For every bus $v \notin V_G^{(k)}$, the aggregated edge weight from $v^*_k$ to $v$ is
\begin{equation}
w'(v^*_k,\, v) = \sum_{u \in V_G^{(k)}} w(u, v),
\label{eq:agg_weight}
\end{equation}
where $w(u,v)=0$ whenever $(u,v)\notin E$, yielding the reduced graph $G'=(V',E',w')$.

\begin{lemma}[Cut preservation]
\label{lem:cut_preserve}
The cut value of any partition on $G'$ equals the cut value of the corresponding expanded partition on $G$.
\end{lemma}
\begin{proof}
Define the surjection $\varphi\colon V\to V'$ by $\varphi(u)=v^*_k$ if $u\in V_G^{(k)}$ for some $k\in\mathcal{K}_c$, and $\varphi(u)=u$ otherwise. Given a partition $\pi\colon V'\to\{1,\ldots,K\}$ on the reduced graph, define the expanded assignment $\tilde{\pi}\colon V\to\{1,\ldots,K\}$ by $\tilde{\pi}(u)=\pi(\varphi(u))$. The cut value on $G$ under $\tilde{\pi}$ is
\begin{align}
C(G,\tilde{\pi})
&= \sum_{(u,v)\in E} w(u,v)\,\mathbf{1}\!\left[\tilde{\pi}(u)\neq\tilde{\pi}(v)\right] \notag\\
&= \sum_{(u,v)\in E} w(u,v)\,\mathbf{1}\!\left[\pi(\varphi(u))\neq\pi(\varphi(v))\right].
\label{eq:cut_expanded}
\end{align}
For any edge $(u,v)\in E$ with $\varphi(u)=\varphi(v)$, both endpoints map to the same super-bus, so $\pi(\varphi(u))=\pi(\varphi(v))$ and the indicator is zero. Hence only edges with $\varphi(u)\neq\varphi(v)$ contribute. Grouping the remaining terms by their image pair $(a,b)=(\varphi(u),\varphi(v))\in E'$ and applying~(\ref{eq:agg_weight}) gives
\begin{align}
C(G,\tilde{\pi})
&= \sum_{(a,b)\in E'}\mathbf{1}\!\left[\pi(a)\neq\pi(b)\right]
   \!\!\sum_{\substack{u\in\varphi^{-1}(a),\,v\in\varphi^{-1}(b)\\(u,v)\in E}}\!\! w(u,v) \notag\\
&= \sum_{(a,b)\in E'} w'(a,b)\,\mathbf{1}\!\left[\pi(a)\neq\pi(b)\right]
 = C(G',\pi). \notag
\end{align}
\end{proof}

By Lemma~\ref{lem:cut_preserve}, the islanding problem on $G'$ shares the same optimal cut value as on $G$. The number of assignment qubits saved by the reduction is
\begin{equation}
\Delta_Q = K\sum_{k \in \mathcal{K}_c}\bigl(|V_G^{(k)}| - 1\bigr),
\label{eq:qubit_saving}
\end{equation}
since a connected coherent group with $|V_G^{(k)}|$ generator buses is represented by one super-bus, eliminating $|V_G^{(k)}|-1$ bus-assignment variables for each island index. This count is independent of the minimum generator and load requirements $N_{G,\min}$ and $N_{L,\min}$ in~(\ref{con:gen})--(\ref{con:load}); those parameters only define the lower-bound constraints. When the reduced graph is used with general $N_{G,\min}$, each super-bus $v^*_k$ retains the generator multiplicity $|V_G^{(k)}|$ for evaluating generator-count feasibility. In addition, buses within each merged group are co-islanded by construction, so the intra-group terms of $H_{\mathrm{coh}}$ in~(\ref{eq:qubo_coh}) vanish identically on $G'$, reducing the Hamiltonian complexity without weakening any feasibility guarantee. After QAOA measurement on $G'$, the solution is expanded back to $G$ by assigning every bus in $V_G^{(k)}$ to the island of its super-bus $v^*_k$.

\subsection{QUBO Conversion of the Objective Function}

To embed the formulation into QAOA, we first rewrite it as a QUBO problem. For each bus $i$, a binary variable $y_{i,k}$ is defined to denote the assignment of that bus to island $k$, where $y_{i,k}=1$ denotes that bus $i$ belongs to island $k$; otherwise, $y_{i,k}=0$. Under this encoding, and assuming each bus is assigned to exactly one island, a line $e_{ij}$ is internal to an island if buses $i$ and $j$ are assigned to the same island, namely when $\sum_{k=1}^{K} y_{i,k}y_{j,k}=1$, and it is cut otherwise. Therefore, the cut objective can be rewritten in binary form as
\begin{equation}
H_{\mathrm{cut}}
=\sum_{e_{ij}\in E} w_{ij}\left(1-\sum_{k=1}^{K} y_{i,k}y_{j,k}\right),\qquad
w_{ij}:=\frac{|p_{ij}|+|p_{ji}|}{2}.
\end{equation}
The one-hot assignment condition required by (\ref{con:onehot}) is then enforced through
\begin{equation}
\sum_{k=1}^{K} y_{i,k}=1,\qquad \forall i\in V.
\end{equation}
To embed the full problem onto QAOA circuits, the remaining inequality constraints are incorporated through a penalty-based reformulation: each constraint is converted into a nonnegative penalty term that is zero when satisfied and strictly positive when violated, and the original cost is augmented by weighted penalties so that any infeasible assignment incurs an energy increase large enough to be dominated in the minimization.

For (\ref{con:onehot})-(\ref{con:coherency}), the QUBO conversion of these constraints can be denoted as:

\begin{equation}
H_h
=\lambda_h\sum_{i\in V}\left(\sum_{k=1}^{K} y_{i,k}-1\right)^2 .
\label{eq:qubo_onehot}
\end{equation}

\begin{equation}
H_s
=\lambda_{M}\sum_{k=1}^{K}\left(\sum_{i\in V} y_{i,k}-N_{\min}-s^{(M)}_{k}\right)^2 .
\label{eq:qubo_minsize}
\end{equation}

\begin{equation}
H_{\mathrm{gen}}
=\lambda_{G}\sum_{k=1}^{K}\left(\sum_{g\in V_G} y_{g,k}-N_{G,\min}-s^{(G)}_{k}\right)^2 .
\label{eq:qubo_gen}
\end{equation}

\begin{equation}
H_{\mathrm{load}}
=\lambda_{L}\sum_{k=1}^{K}\left(\sum_{\ell\in V_L} y_{\ell,k}-N_{L,\min}-s^{(L)}_{k}\right)^2 .
\label{eq:qubo_load}
\end{equation}

\begin{align}
H_{\mathrm{coh}}
&=\lambda_{\mathrm{coh}}\sum_{c=1}^{K}\sum_{k=1}^{K}\sum_{\substack{g,g'\in V_G^{(c)}\\ g<g'}}\left(y_{g,k}-y_{g',k}\right)^2 \notag\\
&\quad+\lambda_{\mathrm{coh}}\sum_{k=1}^{K}\sum_{\substack{c,c'\in\{1,\dots,K\}\\ c<c'}}\ \sum_{g\in V_G^{(c)}}\sum_{g'\in V_G^{(c')}} y_{g,k}y_{g',k}.
\label{eq:qubo_coh}
\end{align}

\noindent where the nonnegative slack variables are encoded by auxiliary binary variables as
\begin{equation}
\begin{split}
s^{(M)}_{k}=\sum_{b=0}^{B_M-1}2^{b}\,u^{(M)}_{k,b},\qquad u^{(M)}_{k,b}\in\{0,1\},\\
B_M=\left\lceil \log_2(N-N_{\min}+1)\right\rceil,
\end{split}
\label{eq:slack_M}
\end{equation}
\begin{equation}
\begin{split}
s^{(G)}_{k}=\sum_{b=0}^{B_G-1}2^{b}\,u^{(G)}_{k,b},\qquad u^{(G)}_{k,b}\in\{0,1\},\\
B_G=\left\lceil \log_2(|V_G|-N_{G,\min}+1)\right\rceil,
\end{split}
\label{eq:slack_G}
\end{equation}
\begin{equation}
\begin{split}
s^{(L)}_{k}=\sum_{b=0}^{B_L-1}2^{b}\,u^{(L)}_{k,b},\qquad u^{(L)}_{k,b}\in\{0,1\},\\
B_L=\left\lceil \log_2(|V_L|-N_{L,\min}+1)\right\rceil.
\end{split}
\label{eq:slack_L}
\end{equation}
These slack variables convert the inequality constraints in (\ref{con:min_size})--(\ref{con:load}) into equality constraints amenable to squared QUBO penalties. The resulting auxiliary binaries $\{u^{(\cdot)}_{k,b}\}$ introduce additional qubits beyond the $NK$ assignment qubits $\{y_{i,k}\}$: $K B_M$ qubits for (\ref{con:min_size}), $K B_G$ qubits for (\ref{con:gen}), and $K B_L$ qubits for (\ref{con:load}). The exact auxiliary-qubit counts used in the experiments are reported in Table~\ref{tab:exp_setup}.

For (\ref{con:connect}), a Depth-First Search (DFS) strategy is used to determine whether the induced subgraphs $G[V_k]$ are connected (i.e., $\mathcal{C}(G[V_k])=1$) \cite{tarjan1972depth}. However, DFS is an algorithmic graph procedure and does not admit a QUBO representation. Therefore, during the QUBO construction for QAOA, we incorporate a quadratic connectivity surrogate that encourages each island to form a single connected component. Specifically, the surrogate penalizes each bus in island $k$ whose neighbors are all assigned to other islands, discouraging isolated nodes and promoting local connectivity. Since exact connectivity admits no compact polynomial binary representation, this quadratic penalty is directly QUBO-compatible and guides the search toward connected partitions without increasing the problem size. After the QAOA measurement, DFS is applied in post-processing to explicitly and exactly verify $\mathcal{C}(G[V_k])=1$ for each sampled solution. The connectivity surrogates can be denoted as:
\begin{equation}
\label{eq:H_iso}
H_c
= \mu \sum_{k=1}^{K}\sum_{i\in V}
\left(
y_{i,k}-\sum_{j\in\mathcal{N}(i)} y_{i,k}\,y_{j,k}
\right),
\end{equation}
\noindent where $\mathcal{N}(i)$ denotes the neighbor set of node $i$.

Thus, the overall problem Hamiltonian for QAOA is constructed by augmenting the objective Hamiltonian with quadratic penalty terms for the physical and operational constraints. Each penalty term is scaled by a positive weight to ensure that any violation increases the total energy, so that low-energy states correspond to feasible islanding solutions. Accordingly, the problem Hamiltonian is given by:
\begin{equation}
\begin{split}
H_{\mathrm{prob}}
={}&H_{\mathrm{cut}}
+\lambda_h\bar{H}_h
+\lambda_{M}\bar{H}_s
+\lambda_{G}\bar{H}_{\mathrm{gen}}\\
&+\lambda_{L}\bar{H}_{\mathrm{load}}
+\lambda_{\mathrm{coh}}\bar{H}_{\mathrm{coh}}
+\mu\,\bar{H}_c.
\end{split}
\label{eq:H_prob}
\end{equation}
where $\lambda_h,\lambda_M,\lambda_G,\lambda_L,\lambda_{\mathrm{coh}},\mu>0$ are penalty weights and $\bar{H}_{(\cdot)}$ denotes the corresponding unweighted QUBO penalty expressions in (\ref{eq:qubo_onehot})--(\ref{eq:qubo_coh}) and (\ref{eq:H_iso}).

\subsection{QAOA Implementation}
Given the problem Hamiltonian $H_{\mathrm{prob}}$ in (\ref{eq:H_prob}), the next step is to construct a QAOA circuit that approximately minimizes its expected energy. To implement $H_{\mathrm{prob}}$ on quantum hardware, we first rewrite the binary variables into Pauli-$Z$ operators via the standard mapping
\begin{equation}
\label{eq:bin_qo_ising}
z_{i,k}=1-2y_{i,k},\qquad y_{i,k}\in\{0,1\},\ z_{i,k}\in\{\pm1\},
\end{equation}
which transforms the QUBO energy into an Ising-form cost Hamiltonian. 

The standard transverse-field initialization Hamiltonian is adopted to promote exploration over the computational basis:
\begin{equation}
\label{eq:H_B}
H_B=\sum_{q=1}^{n_q}\sigma_q^{x},
\end{equation}
and initialize the circuit in the uniform superposition state
\begin{equation}
\label{eq:psi0}
|\psi_0\rangle=\bigotimes_{q=1}^{n_q}\frac{|0\rangle+|1\rangle}{\sqrt{2}}.
\end{equation}
With a $p$-layer QAOA, the variational state is prepared as
\begin{equation}
\label{eq:qaoa_state}
|\psi(\boldsymbol{\gamma},\boldsymbol{\beta})\rangle
=\prod_{\ell=1}^{p} e^{-i\beta_\ell H_B}\,e^{-i\gamma_\ell H_{\mathrm{prob}}}\,|\psi_0\rangle,
\end{equation}
where $\boldsymbol{\gamma}=(\gamma_1,\dots,\gamma_p)$ and $\boldsymbol{\beta}=(\beta_1,\dots,\beta_p)$ are variational parameters updated by a classical outer loop. For each candidate $(\boldsymbol{\gamma},\boldsymbol{\beta})$, the QAOA circuit is executed and measured to obtain a batch of bitstrings $z\sim P_{\boldsymbol{\gamma},\boldsymbol{\beta}}$, which are decoded into partitions $\{V_k(z)\}$. Accordingly, instead of filtering out DFS-infeasible samples, we incorporate the DFS connectivity outcome directly into the outer-loop objective by augmenting the measured energy with a feasibility penalty. Specifically, for a measured bitstring $z$ (decoded to $\{V_k(z)\}$), define the DFS indicator
\begin{equation}
\label{eq:dfs_indicator}
\chi_C(z)
:=\prod_{k=1}^{K}\mathbb{I}\!\left\{\mathcal{C}\!\left(G[V_k(z)]\right)=1\right\}\in\{0,1\},
\end{equation}
\noindent where $\mathbb{I}\{\cdot\}$ is the indicator function and $\chi_C(z)=1$ iff all induced subgraphs $G[V_k(z)]$ are connected. The corresponding DFS-modified energy is
\begin{equation}
\label{eq:dfs_modified_energy}
\widetilde{H}_{\mathrm{prob}}(z)
:=H_{\mathrm{prob}}(z)+\lambda_C\big(1-\chi_C(z)\big),
\end{equation}
where $\lambda_C>0$ is a penalty weight. The classical outer loop then minimizes the DFS-modified expected energy
\begin{equation}
\label{eq:dfs_modified_objective}
\min_{\boldsymbol{\gamma},\boldsymbol{\beta}}\;
\widetilde{\mathbb{E}}(\boldsymbol{\gamma},\boldsymbol{\beta})
:=\mathbb{E}_{z\sim P_{\boldsymbol{\gamma},\boldsymbol{\beta}}}
\!\left[\widetilde{H}_{\mathrm{prob}}(z)\right],
\end{equation}
which can be estimated from $S$ circuit executions as
\begin{equation}
\label{eq:dfs_modified_estimator}
\widehat{\widetilde{\mathbb{E}}}(\boldsymbol{\gamma},\boldsymbol{\beta})
=\frac{1}{S}\sum_{s=1}^{S}\widetilde{H}_{\mathrm{prob}}\!\left(z^{(s)}\right),
\qquad z^{(s)}\sim P_{\boldsymbol{\gamma},\boldsymbol{\beta}},
\end{equation}
where $P_{\boldsymbol{\gamma},\boldsymbol{\beta}}$ denotes the QAOA measurement distribution induced by the variational state. In addition, this work utilizes COBYLA as the classical optimizer for QAOA circuit parameter updating. However, the output of the variational loop is still a distribution over measured bitstrings rather than a guaranteed feasible islanding decision. Therefore, after the QAOA sampling stage, an additional classical post-processing layer is introduced to decode, evaluate, and repair promising samples before the final islanding solution is selected.

\section{Post-Processing Framework for Feasible Islanding}
\label{sec:postprocessing}

Following the QAOA execution and measurement stage in Fig.~\ref{fig:overall_diagram}~(c), measured bitstrings are passed to the classical post-processing operation shown in Fig.~\ref{fig:overall_diagram}~(d), and the selected result is decoded into the islanding solution shown in Fig.~\ref{fig:overall_diagram}~(e). Because the islanding constraints sharply restrict the feasible region, many low-energy samples remain infeasible. Post-processing is therefore used to repair promising samples and improve the final solution quality and feasibility.

Related work \cite{shirai2024postprocessing} studied QUBO-based greedy post-processing that maps quantum outputs to constraint-satisfying solutions under monotone energy descent. For islanding, this suggests starting with a QUBO-only repair rule, since the cut objective and the algebraic assignment, size, generator, and load constraints are already expressed in QUBO form. Exact graph connectivity, however, is verified by DFS rather than by a native QUBO term. This naturally motivates a staged repair strategy that first handles the QUBO constraints and then connectivity. Because such a staged repair can improve DFS feasibility while increasing the QUBO penalty, the later methods introduce stricter descent controls, modified repair targets, or relaxed admissible moves. The five methods proposed below develop that basic idea with different feasibility guarantees.

\subsection{Post-Processing Mathematical Preliminaries}
\label{sec:post_common}

All the methods operate on the same mathematical objects and share the same elementary operations. A measured bitstring $z \in \{0,1\}^{NK}$, indexing all binary bus-to-island assignment variables ($N$ buses, $K$ islands), is decoded into an island partition $\{V_k(z)\}_{k=1}^{K}$. Each method evaluates the QUBO energy and a connectivity-violation objective, and applies single-variable flips as its local move. For a decoded configuration $z$, the QUBO energy is
\begin{equation}
\begin{aligned}
Q(z) ={}& H_{\mathrm{cut}}(z)
+\lambda_h\bar{H}_h(z) \\
&+\lambda_M\bar{H}_s(z)
+\lambda_G\bar{H}_{\mathrm{gen}}(z) \\
&+\lambda_L\bar{H}_{\mathrm{load}}(z)
+\lambda_{\mathrm{coh}}\bar{H}_{\mathrm{coh}}(z).
\end{aligned}
\label{eq:qubo_only_post}
\end{equation}
which combines the cut objective with the algebraic constraint penalties. The connectivity-violation objective is
\begin{equation}
C(z) := \sum_{k=1}^{K}\bigl(\omega(G[V_k(z)]) - 1\bigr),
\label{eq:qdfs_post}
\end{equation}
where $\omega(G[V_k(z)])$ denotes the number of connected components in island $k$. Thus, $C(z)=0$ if and only if every island is internally connected.

Each algebraic constraint in~(\ref{eq:qubo_only_post}) is encoded as a nonnegative penalty. For constraint $m$ with coefficient vector $a^m_i \in \mathbb{Z}$, feasibility bounds $[b^m_{\min}, b^m_{\max}]$, and constraint sum $S^m(z) = \sum_{i \in \mathcal{V}^m} a^m_i\, z_i$, the general hinge-loss form is
\begin{equation}
  H^m(z)
  \;=\;
  \max\!\Bigl(0,\; S^m(z) - b^m_{\max},\; b^m_{\min} - S^m(z)\Bigr),
  \label{eq:hinge_penalty}
\end{equation}
which equals zero when $S^m(z) \in [b^m_{\min}, b^m_{\max}]$ and is strictly positive otherwise. In the QUBO of Section~\ref{quantum}, inequality constraints are converted to equalities via binary slack variables~(\ref{eq:slack_M})--(\ref{eq:slack_L}), so each penalty takes the equivalent squared form $H^m(z) = (S^m(z) - b^m)^2$ with $b^m = b^m_{\min} = b^m_{\max}$ after substitution. The QUBO energy therefore decomposes as $Q(z) = H_{\mathrm{cut}}(z) + \sum_{m} \lambda_m H^m(z)$.

For variable index $j$, the flip $z \oplus e_j$ denotes the configuration with $z_j$ toggled and all other entries unchanged. For any energy function $E$, the per-flip change is
\begin{equation}
  \Delta^E_j(z) \;:=\; E(z \oplus e_j) - E(z).
  \label{eq:delta_generic}
\end{equation}
When $E = Q$ we write $\Delta_j(z)$ for brevity; when $E = C$ we write $\Delta^C_j(z)$.

All methods that repair QUBO constraints share the same Stage~1 subroutine: starting from the current configuration $z^{(t)}$, select
\begin{equation}
  j^*(t) \;=\; \operatorname*{arg\,min}_{j}\;\Delta_j\!\left(z^{(t)}\right),
  \label{eq:greedy_select}
\end{equation}
and update
\begin{equation}
  z^{(t+1)} \;=\;
  \begin{cases}
    z^{(t)} \oplus e_{j^*(t)} & \text{if } \Delta_{j^*(t)}\!\left(z^{(t)}\right) < 0, \\
    z^{(t)} & \text{(terminate).}
  \end{cases}
  \label{eq:greedy_flip}
\end{equation}
Because each flip strictly decreases $Q$ on a finite state space, Stage~1 is cycle-free and terminates in at most $2^N - 1$ steps. Under condition~(A3) stated in Appendix~\ref{app:postprocessing}, every local minimum of $Q$ is QUBO-feasible.

Methods that include a connectivity-repair stage (Methods~2, 3, and~4) operate on the set of flips that strictly reduce $C$,
\begin{equation}
  \mathcal{J}(z)
  \;=\;
  \bigl\{\, j : \Delta^C_j(z) < 0 \bigr\}.
  \label{eq:jconn}
\end{equation}
The methods differ in what additional condition, if any, is imposed on the $Q$ change $\Delta_{j}(z)$ for candidates in $\mathcal{J}(z)$.

\subsection{M.~1: QUBO-only descent}
\label{sec:method1}

M.~1 is exactly the Stage~1 greedy descent subroutine from Section~\ref{sec:post_common} applied to $Q$, with no subsequent repair stage. It is the direct islanding adaptation of the post-measurement greedy repair in \cite{shirai2024postprocessing}. As shown in Algorithm~\ref{alg:m1}, the algorithm runs~(\ref{eq:greedy_select})--(\ref{eq:greedy_flip}) until no flip decreases $Q$, then performs a single DFS pass to evaluate $C(z^*)$. Termination follows directly from Section~\ref{sec:post_common}; QUBO-feasibility of the output holds under condition~(A3) stated in Appendix~\ref{app:postprocessing}. Since $Q$ contains no connectivity term, however, $C(z^*)$ may be nonzero and connectivity of the resulting partition is not guaranteed.

\begin{algorithm}[!t]
\caption{M.1: QUBO-Only Descent}
\label{alg:m1}
\begin{algorithmic}[1]
\State \textbf{Data:} bitstring $z \in \{0,1\}^{NK}$, QUBO energy $Q$
\State \textbf{Result:} $z^*$ at a local minimum of $Q$; $z^* \in \mathcal{F}_Q$ under (A3)
\While{$\exists\, j$ such that $\Delta_j(z) < 0$}
    \State $j^* \leftarrow \operatorname*{arg\,min}_{j}\;\Delta_j(z)$
    \State $z \leftarrow z \oplus e_{j^*}$
\EndWhile
\State $z^* \leftarrow z$
\State Evaluate $C(z^*)$ via DFS \Comment{connectivity not guaranteed}
\State \Return $z^*$
\end{algorithmic}
\end{algorithm}

\subsection{M.~2: Unconstrained two-stage repair}
\label{sec:method2}

M.~2 appends a connectivity-repair stage to M.~1. After the Stage~1 descent from Section~\ref{sec:post_common} terminates, Stage~2 accepts any reassignment that reduces $C$, regardless of the effect on $Q$.

After Stage~1 terminates, if $C(z) > 0$ then Stage~2 operates on $\mathcal{J}(z)$ from~(\ref{eq:jconn}) and picks
\begin{equation}
  j^* \;=\; \operatorname*{arg\,min}_{j \in \mathcal{J}(z)}\; \Delta_j(z),
  \label{eq:m2_stage2_select}
\end{equation}
where $\Delta_j(z)$ is the per-flip $Q$ change from~(\ref{eq:delta_generic}). The flip $z \leftarrow z \oplus e_{j^*}$ is executed unconditionally, without checking the sign of $\Delta_{j^*}$, and Stage~1 is restarted on the updated configuration. Algorithm~\ref{alg:m2} shows this alternation, which repeats for at most $K$ outer iterations.

\begin{remark}
The Stage~1 descent still strictly decreases $Q$. Stage~2, however, imposes no descent condition: a connectivity-improving flip can simultaneously increase the QUBO penalty terms, so $Q$ may rise after a Stage~2 move. The two stages therefore optimize different targets in sequence, and no globally consistent energy function is monotonically decreasing throughout the run. As a consequence, the same configuration can be revisited across outer iterations, and an outer bound $K$ is required for termination to be guaranteed.
\end{remark}

\begin{algorithm}[!t]
\caption{M.2: Unconstrained Two-Stage Repair}
\label{alg:m2}
\begin{algorithmic}[1]
\State \textbf{Data:} bitstring $z \in \{0,1\}^{NK}$, energies $Q$ and $C$, outer bound $K$
\State \textbf{Result:} $z^*$ with $Q$ locally minimal; $C(z^*)=0$ if found within $K$ iterations
\For{$k = 1$ \textbf{to} $K$}
    \While{$\exists\, j$ such that $\Delta_j(z) < 0$} \Comment{Stage 1: QUBO descent}
        \State $j^* \leftarrow \operatorname*{arg\,min}_{j}\;\Delta_j(z)$
        \State $z \leftarrow z \oplus e_{j^*}$
    \EndWhile
    \If{$C(z) = 0$}
        \State \Return $z$ \Comment{connectivity satisfied}
    \EndIf
    \State $j^* \leftarrow \operatorname*{arg\,min}_{j \in \mathcal{J}(z)}\;\Delta_j(z)$ \Comment{Stage 2: unconstrained repair}
    \State $z \leftarrow z \oplus e_{j^*}$ \Comment{flip regardless of sign of $\Delta_{j^*}$}
\EndFor
\State \Return $z$
\end{algorithmic}
\end{algorithm}

\subsection{M.~3: Constrained two-stage descent}
\label{sec:method3}

M.~3 retains the two-stage structure of M.~2 and strengthens Stage~2 by admitting a connectivity-repair flip only if it simultaneously strictly decreases $Q$.

After the Stage~1 descent from Section~\ref{sec:post_common} terminates with $C(z) > 0$, Stage~2 forms $\mathcal{J}(z)$ as in~(\ref{eq:jconn}) and selects the best candidate,
\begin{equation}
  j^* \;=\; \operatorname*{arg\,min}_{j \in \mathcal{J}(z)}\; \Delta_j(z).
\end{equation}
The flip is then conditionally applied:
\begin{equation}
  z \;\leftarrow\;
  \begin{cases}
    z \oplus e_{j^*} & \text{if } \Delta_{j^*}(z) < 0 \text{ (restart Stage~1)}, \\
    z & \text{if } \Delta_{j^*}(z) \geq 0 \text{ (terminate)}.
  \end{cases}
  \label{eq:m3_stage2}
\end{equation}
Relative to M.~2, the only modification is the sign condition on $\Delta_{j^*}$ in~(\ref{eq:m3_stage2}), which ensures that $Q$ decreases at every executed flip and thereby eliminates the need for an outer iteration bound $K$. The resulting constrained two-stage descent procedure is shown in Algorithm~\ref{alg:m3}.

\begin{remark}
Because both stages flip only when $\Delta_{j^*}(z) < 0$, every executed move strictly decreases $Q$. Strict monotone decrease on a finite state space excludes revisits and bounds the total flip count without any outer iteration limit. Under condition~(A3) stated in Appendix~\ref{app:postprocessing}, Stage~1 local minima are QUBO-feasible. The failure case in~(\ref{eq:m3_stage2}) arises when every connectivity-improving flip violates a tight equality constraint, causing $\Delta_{j^*}(z) \geq 0$ for all $j^* \in \mathcal{J}$. Formal proofs of global strict descent, the no-cycle property, and the tight-constraint failure mode are given in Appendix~\ref{app:postprocessing}.
\end{remark}

\begin{algorithm}[!t]
\caption{M.3: Constrained Two-Stage Descent}
\label{alg:m3}
\begin{algorithmic}[1]
\State \textbf{Data:} bitstring $z \in \{0,1\}^{NK}$, energies $Q$ and $C$
\State \textbf{Result:} $z^* \in \mathcal{F}_Q$ under (A3); $C(z^*)=0$ if a descent-compatible repair exists
\While{\textbf{true}}
    \While{$\exists\, j$ such that $\Delta_j(z) < 0$} \Comment{Stage 1: QUBO descent}
        \State $j^* \leftarrow \operatorname*{arg\,min}_{j}\;\Delta_j(z)$
        \State $z \leftarrow z \oplus e_{j^*}$
    \EndWhile
    \If{$C(z) = 0$}
        \State \Return $z$
    \EndIf
    \State $j^* \leftarrow \operatorname*{arg\,min}_{j \in \mathcal{J}(z)}\;\Delta_j(z)$ \Comment{Stage 2: constrained repair}
    \If{$\Delta_{j^*}(z) < 0$}
        \State $z \leftarrow z \oplus e_{j^*}$ \Comment{restart Stage 1}
    \Else
        \State \Return $z$ \Comment{no descent-compatible repair exists}
    \EndIf
\EndWhile
\end{algorithmic}
\end{algorithm}

\subsection{M.~4: Penalty-relaxation descent}
\label{sec:method4}

M.~4 keeps the two-stage structure of M.~3 but introduces a per-iteration penalty schedule that gradually reduces the QUBO penalty weight across outer iterations, enlarging the set of admissible Stage~2 flips.

For outer iteration $r \geq 1$, fix a non-decreasing function $f\colon \mathbb{N}^+ \to [1,\infty)$ with $f(1)=1$ and define the iteration-$r$ penalty coefficient and energy,
\begin{equation}
  \lambda^{(r)} \;=\; \frac{\lambda}{f(r)},
  \qquad
  Q^{(r)}(z) \;=\; H_{\mathrm{cut}}(z) \;+\; \lambda^{(r)} \sum_{m} H^m(z),
  \label{eq:m4_relaxed}
\end{equation}
where $\lambda = \lambda^{(1)}$ is the base coefficient used by M.~1 through M.~3. A concrete implementation sets $f(r) = 1 + \rho(r-1)$ for a relaxation step $\rho > 0$, so that $\lambda^{(r)}$ decreases linearly and $Q^{(r)}$ changes across iterations.

Both stages in outer iteration $r$ use $Q^{(r)}$ and its per-flip change $\Delta^{(r)}_j(z) = Q^{(r)}(z \oplus e_j) - Q^{(r)}(z)$ consistently. Stage~1 applies the same greedy rule as~(\ref{eq:greedy_flip}) with $\Delta_j$ replaced by $\Delta^{(r)}_j$. Stage~2 then seeks $j^* \in \mathcal{J}(z)$ satisfying
\begin{equation}
  \Delta^{(r)}_{j^*}(z) \;<\; 0.
  \label{eq:m4_stage2_cond}
\end{equation}
If~(\ref{eq:m4_stage2_cond}) holds, the flip is executed and Stage~1 is restarted on the updated configuration at the same iteration $r$. If no such $j^*$ exists, $r$ is incremented, which reduces $\lambda^{(r+1)}$, and Stage~1 restarts. Algorithm~\ref{alg:m4} shows this penalty-relaxation procedure, which terminates on success (DFS satisfied) or when $r$ exceeds the outer bound $R_{\max}$.

\begin{remark}
A smaller $\lambda^{(r)}$ lowers the threshold in~(\ref{eq:m4_stage2_cond}), so a flip that repairs connectivity at the cost of a QUBO violation is more likely to yield $\Delta^{(r)}_{j^*} < 0$. However, the QUBO feasibility guarantee of Stage~1 requires $\lambda^{(r)} > \Delta_{\max} = \max_{z,j}|\Delta^{H_{\mathrm{cut}}}_j(z)|$. If $\lambda^{(r)}$ falls below this threshold, Stage~1 may terminate at a QUBO-infeasible configuration. The two requirements are therefore in fundamental tension, and no single $\lambda^{(r)}$ simultaneously ensures both. Within each fixed $r$ the descent on $Q^{(r)}$ is strict, and no configuration is revisited. Across iterations, cycles are possible when $f$ is constant, so an outer bound $R_{\max}$ is required for termination. Formal proofs of within-iteration descent, the cross-iteration cycle condition, and the QUBO-feasibility incompatibility are given in Appendix~\ref{app:postprocessing}.
\end{remark}

\begin{algorithm}[!t]
\caption{M.4: Penalty-Relaxation Descent}
\label{alg:m4}
\begin{algorithmic}[1]
\State \textbf{Data:} bitstring $z \in \{0,1\}^{NK}$, base penalty $\lambda$, relaxation schedule $f(r)$, bound $R_{\max}$
\State \textbf{Result:} $z^*$ with $C(z^*)=0$ if found within $R_{\max}$ iterations
\For{$r = 1$ \textbf{to} $R_{\max}$}
    \State $\lambda^{(r)} \leftarrow \lambda / f(r)$; \quad $Q^{(r)}(z) \leftarrow H_{\mathrm{cut}}(z) + \lambda^{(r)}\sum_m H^m(z)$
    \State $\mathit{repaired} \leftarrow \textbf{true}$
    \While{$\mathit{repaired}$}
        \While{$\exists\, j$ such that $\Delta^{(r)}_j(z) < 0$} \Comment{Stage 1: descent on $Q^{(r)}$}
            \State $j^* \leftarrow \operatorname*{arg\,min}_{j}\;\Delta^{(r)}_j(z)$
            \State $z \leftarrow z \oplus e_{j^*}$
        \EndWhile
        \If{$C(z) = 0$}
            \State \Return $z$
        \EndIf
        \If{$\exists\, j^* \in \mathcal{J}(z)$ such that $\Delta^{(r)}_{j^*}(z) < 0$} \Comment{Stage 2: relaxed repair}
            \State $j^* \leftarrow \operatorname*{arg\,min}_{j \in \mathcal{J}(z)}\;\Delta^{(r)}_j(z)$
            \State $z \leftarrow z \oplus e_{j^*}$ \Comment{restart Stage 1 at same $r$}
        \Else
            \State $\mathit{repaired} \leftarrow \textbf{false}$ \Comment{no admissible flip; increment $r$}
        \EndIf
    \EndWhile
\EndFor
\State \Return $z$
\end{algorithmic}
\end{algorithm}

\subsection{M.~5: Unified-energy descent}
\label{sec:method5}

M.~5 eliminates the two-stage structure entirely by encoding both QUBO feasibility and DFS connectivity into a single post-processing energy,
\begin{equation}
\mathcal{Q}(z)
\;:=\;
Q(z)
\;+\; \mu\, C(z),
\label{eq:unified_post}
\end{equation}
and applying the common greedy descent directly to $\mathcal{Q}$.

For the unified energy to guarantee full feasibility, $\mu$ must dominate all QUBO-side effects of any single flip. Concretely, the required condition is
\begin{equation}
  \mu
  \;>\;
  \max_{z,\,j}\bigl|\Delta_j(z)\bigr|,
  \label{eq:m5_hierarchy}
\end{equation}
where $\Delta_j(z)$ is the per-flip QUBO change defined in Section~\ref{sec:post_common}. Under~(\ref{eq:m5_hierarchy}), any flip that reduces $C$ produces a net decrease in $\mathcal{Q}$ regardless of its QUBO effect, so connectivity repair is always preferred over any QUBO trade-off.

The per-flip change in $\mathcal{Q}$ decomposes as
\begin{equation}
  \Delta^u_j(z)
  \;=\;
  \Delta_j(z)
  \;+\;
  \mu\,
  \Delta^C_j(z),
  \label{eq:delta_uni}
\end{equation}
using the QUBO and DFS per-flip changes already introduced in Section~\ref{sec:post_common}. M.~5 then iterates the single rule
\begin{equation}
  j^*(t) \;=\; \operatorname*{arg\,min}_{j}\;\Delta^u_j\!\left(z^{(t)}\right),
  \label{eq:m5_select}
\end{equation}
and updates the configuration by
\begin{equation}
  z^{(t+1)} =
  \begin{cases}
    z^{(t)} \oplus e_{j^*(t)} & \text{if } \Delta^u_{j^*(t)} < 0, \\
    z^{(t)} & \text{if } \Delta^u_{j^*(t)} \geq 0 \text{ (terminate)}.
  \end{cases}
  \label{eq:m5_flip}
\end{equation}
The QUBO and connectivity objectives are thus handled within a single unified descent loop, with no alternation between stages. The complete unified-energy descent procedure is shown in Algorithm~\ref{alg:m5}.

\begin{remark}
Strict descent on the fixed function $\mathcal{Q}$ makes the trajectory cycle-free and guarantees termination without an outer bound $K$. Under condition~(\ref{eq:m5_hierarchy}), conditions~(A4) and~(A5), and the flip-existence condition of Appendix~\ref{app:postprocessing}, every local minimum satisfies $C(z^*)=0$ and all algebraic constraints, so $z^* \in \mathcal{F}_Q \cap \mathcal{F}_C$; M.~5 thus provides the strongest feasibility guarantee of the five methods, which directly resolves the failure case of M.~3, where the two-stage separation can leave every connectivity-improving flip QUBO-increasing and force Stage~2 to terminate without a feasible repair. The practical cost is that condition~(\ref{eq:m5_hierarchy}) requires a large $\mu$, which biases every descent step toward connectivity repair regardless of QUBO state. Since QAOA samples tend to be near QUBO-feasible, M.~3 can exploit this: Stage~1 polishes $Q$ with few flips and Stage~2 then performs only a small targeted repair, whereas M.~5 blends both objectives from the outset and may terminate with worse $Q$. Formal proofs are given in Appendix~\ref{app:postprocessing}.
\end{remark}

\begin{algorithm}[!t]
\caption{M.5: Unified-Energy Descent}
\label{alg:m5}
\begin{algorithmic}[1]
\State \textbf{Data:} bitstring $z \in \{0,1\}^{NK}$, energies $Q$ and $C$, weight $\mu > \max_{z,j}|\Delta_j(z)|$
\State \textbf{Result:} $z^* \in \mathcal{F}_Q \cap \mathcal{F}_C$ under (A4), (A5)
\While{$\exists\, j$ such that $\Delta^u_j(z) < 0$} \Comment{single descent on $\mathcal{Q} = Q + \mu C$}
    \State $j^* \leftarrow \operatorname*{arg\,min}_{j}\;\Delta^u_j(z)$
    \Statex \hspace{\algorithmicindent} where $\Delta^u_j(z) = \Delta_j(z) + \mu\,\Delta^C_j(z)$
    \State $z \leftarrow z \oplus e_{j^*}$
\EndWhile
\State \Return $z$ \Comment{$z \in \mathcal{F}_Q \cap \mathcal{F}_C$ under stated conditions}
\end{algorithmic}
\end{algorithm}

\subsection{Method Summary and Interpretation}

M.~1 and M.~2 are natural baselines derived from the QUBO formulation, while M.~3 through M.~5 are principled developments that add formal descent and feasibility guarantees.
\begin{itemize}
    \item M.~1 and M.~2 establish QUBO-only and staged-repair baselines, with no guarantee of $z^*\in\mathcal{F}_Q\cap\mathcal{F}_C$.
    \item M.~3 adds a descent filter to Stage~2 of M.~2, achieving strict global monotonicity on $Q$.
    \item M.~4 relaxes the penalty weight $\lambda^{(r)}$ across iterations to recover Stage~2 flips blocked by M.~3.
    \item M.~5 replaces the staged structure with a single descent on $\mathcal{Q}=Q+\mu C$, and under the regularity conditions stated in Appendix~\ref{app:postprocessing} is the only formulation that guarantees $z^*\in\mathcal{F}_Q\cap\mathcal{F}_C$.
\end{itemize}

Table~\ref{tab:post_methods} summarizes the key properties. Formal proofs are given in Appendix~\ref{app:postprocessing}.

\begin{table}[!t]
\centering
\caption{Summary of the proposed post-processing methods.}
\label{tab:post_methods}
\begin{threeparttable}
\footnotesize
\resizebox{\columnwidth}{!}{%
\begin{tabular}{cllllc}
\toprule
Method & Objective & Descent & Output at $z^*$ & Terminates \\
\midrule
1 & $Q$ & Strict & Conditional $\mathcal{F}_Q$ & Always \\
2 & $Q \to C$ (staged) & Stage~1 & --- & With cutoff \\
3 & $Q$ (two-stage) & Strict & Conditional $\mathcal{F}_Q$ & Always \\
4 & $Q^{(r)}$,\; $\lambda^{(r)}\!\downarrow$ & Per-$r$ & --- & With cutoff \\
5 & $\mathcal{Q}=Q+\mu C$ & Strict & Conditional $\mathcal{F}_Q\cap\mathcal{F}_C$ & Always \\
\bottomrule
\end{tabular}
}
\end{threeparttable}
\end{table}

\section{Case Study}
\label{sec:numerical}

\subsection{Experimental Setup and Test Systems}

The numerical study is conducted on six standard IEEE benchmark systems: 9-bus, 14-bus, 24-bus,
30-bus, 39-bus, and 57-bus. These cases span different network sizes and islanding structures, including both
two-island and three-island settings, and therefore provide a representative test bed for evaluating the proposed
QAOA-based islanding framework under increasing combinatorial complexity. For each case, the coherent generator
groups determine the target number of islands~$K$, and the resulting QAOA samples are assessed both before and after post-processing. To benchmark solution quality,
the post-processed quantum results are compared against vanilla QAOA sampling and against a classical reference
solver Gurobi. The evaluation focuses
on cut quality, quantum runtime, feasibility behavior, and physical operability of the resulting islands.
For all benchmark cases, the minimum island size is set to $N_{\min}=2$, and the generator and load count requirements in~(\ref{con:gen})--(\ref{con:load}) are set to $N_{G,\min}=1$ and $N_{L,\min}=1$, respectively.

The experiments are implemented within the IBM Quantum stack using a real IBM
Quantum backend. In particular, the hardware runs use \texttt{ibm\_marrakesh}, which IBM identifies as a
156-qubit Heron~r2 processor~\cite{ibm_compute_resources_2026}. For quantum implementation, Table~\ref{tab:exp_setup} summarizes the key configuration choices for each system, including the target
island count~$K$, the assignment, auxiliary, and total qubit counts before and after coherency-based bus merging, the QAOA depth parameter~$p$, the
number of shots~$S$, and the total classical iteration budget~$I_{\max}$. The auxiliary-qubit breakdown reflects how the minimum-bus, generator-count, and load-count constraints contribute to the quantum resource requirement. Each configuration is evaluated
over five independent trials, and the reported statistics are the sample mean and standard deviation across
those repetitions.

The framework is further evaluated under realistic hardware noise conditions via the \texttt{FakeMarrakesh} noise model~\cite{ibm_fakemarrakesh_2026}, a calibrated replica of the ibm\_marrakesh Heron~r2 device noise profile. This evaluation is restricted to the IEEE~9-bus case, as its 32-qubit circuit (Table~\ref{tab:exp_setup}) is the only configuration within the 32-qubit ceiling imposed by the IBM cloud simulator~\cite{ibm_compute_resources_2026}, and each condition is repeated over five independent trials with results reported as the sample mean and standard deviation.

\begin{table}[!tbp]
\centering
\scriptsize
\setlength{\tabcolsep}{2.6pt}
\renewcommand{\arraystretch}{1.05}
\caption{QAOA benchmark settings and qubit counts.}
\label{tab:exp_setup}
\begin{tabular}{@{}lcccccccccccccc@{}}
\toprule
\multirow{2}{*}{\textbf{Case}} & \multirow{2}{*}{$K$} & \multicolumn{2}{c}{\makecell{\textbf{Assignment}\\\textbf{qubits}}} & \multicolumn{2}{c}{\makecell{\textbf{Minimum-bus}\\\textbf{auxiliary qubits}}} & \multicolumn{2}{c}{\makecell{\textbf{Generator}\\\textbf{auxiliary qubits}}} & \multicolumn{2}{c}{\makecell{\textbf{Load}\\\textbf{auxiliary qubits}}} & \multicolumn{2}{c}{\makecell{\textbf{Total}\\\textbf{qubits}}} & \multirow{2}{*}{$p$} & \multirow{2}{*}{$S$} & \multirow{2}{*}{$I_{\max}$} \\
\cmidrule(lr){3-4}\cmidrule(lr){5-6}\cmidrule(lr){7-8}\cmidrule(lr){9-10}\cmidrule(lr){11-12}
 & & \makecell{w/o\\merging} & \makecell{w/\\merging} & \makecell{w/o\\merging} & \makecell{w/\\merging} & \makecell{w/o\\merging} & \makecell{w/\\merging} & \makecell{w/o\\merging} & \makecell{w/\\merging} & \makecell{w/o\\merging} & \makecell{w/\\merging} & & & \\
\midrule
9-bus  & 2 & 18  & 18  & 6  & 6  & 4  & 4  & 4  & 4  & 32  & 32  & 1 & 100   & 2 \\
14-bus & 2 & 28  & 24  & 8  & 8  & 6  & 4  & 8  & 8  & 50  & 44  & 1 & 200   & 2 \\
24-bus & 3 & 72  & 60  & 15 & 15 & 12 & 9  & 12 & 12 & 111 & 96  & 2 & 800   & 3 \\
30-bus & 2 & 60  & 58  & 10 & 10 & 6  & 6  & 10 & 10 & 86  & 84  & 2 & 1000  & 2 \\
39-bus & 3 & 117 & 117 & 16 & 16 & 10 & 10 & 13 & 13 & 156 & 156 & 4 & 3000  & 3 \\
57-bus & 2 & 114 & 104 & 12 & 12 & 6  & 2  & 12 & 12 & 144 & 130 & 4 & 3000  & 3 \\
\bottomrule
\end{tabular}
\end{table}

\begin{table}[!tbp]
\centering
\scriptsize
\setlength{\tabcolsep}{2pt}
\renewcommand{\arraystretch}{1.08}
\caption{Bus assignments of the islanding solutions obtained by the proposed QAOA framework.}
\label{tab:islanding_partitions}
\begin{tabular}{@{}>{\raggedright\arraybackslash}p{0.19\columnwidth}
                @{\hspace{0.04\columnwidth}}
                >{\raggedright\arraybackslash}p{0.71\columnwidth}@{}}
\toprule
\makecell[l]{\textbf{IEEE test}\\\textbf{case}} & \textbf{Islanding solution (bus sets)} \\
\midrule
\textbf{9-bus} &
\makecell[l]{Island 1: 1, 4, 5; Island 2: 2, 3, 6--9} \\
\addlinespace[2pt]
\textbf{14-bus} &
\makecell[l]{Island 1: 1--5; Island 2: 6--14} \\
\addlinespace[2pt]
\textbf{24-bus} &
\makecell[l]{Island 1: 6, 10--14, 20, 23; Island 2: 3, 15--19, 21, 22, 24;\\
Island 3: 1, 2, 4, 5, 7, 8, 9} \\
\addlinespace[2pt]
\textbf{30-bus} &
\makecell[l]{Island 1: 9--30; Island 2: 1--8} \\
\addlinespace[2pt]
\textbf{39-bus} &
\makecell[l]{Island 1: 1, 2, 3, 25--30, 37, 38; Island 2: 15--24, 33--36;\\
Island 3: 4--14, 31, 32, 39} \\
\addlinespace[2pt]
\textbf{57-bus} &
\makecell[l]{Island 1: 1--5, 11, 13--23, 32--49, 56, 57;\\
Island 2: 6--10, 12, 24--31, 50--55} \\
\bottomrule
\end{tabular}
\end{table}

\begin{figure}[!t]
\centering
\setlength{\tabcolsep}{2pt}
\renewcommand{\arraystretch}{1.0}
\begin{tabular}{@{}c@{\hspace{4pt}}c@{}}
\includegraphics[width=0.42\textwidth]{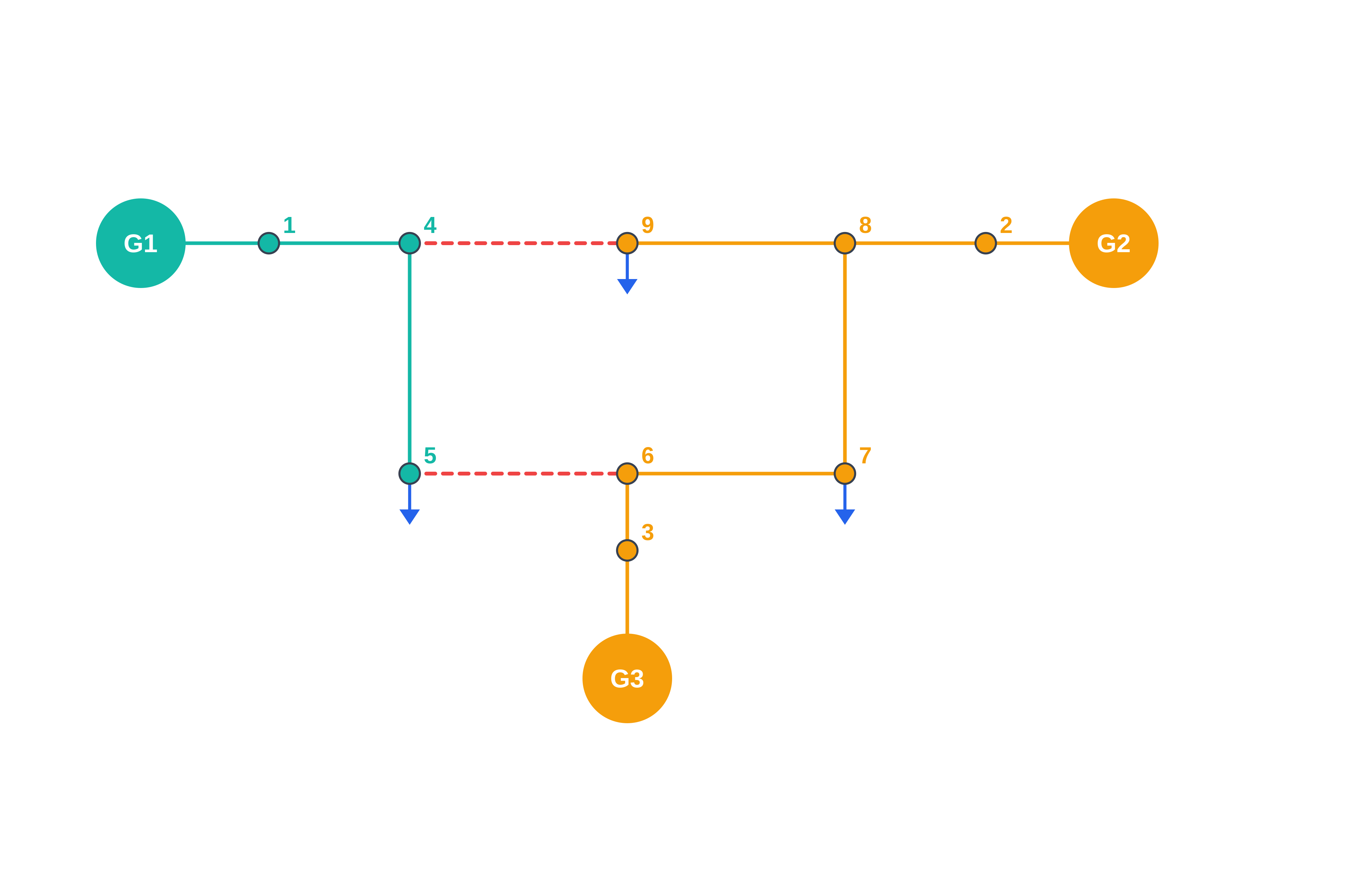} &
\includegraphics[width=0.42\textwidth]{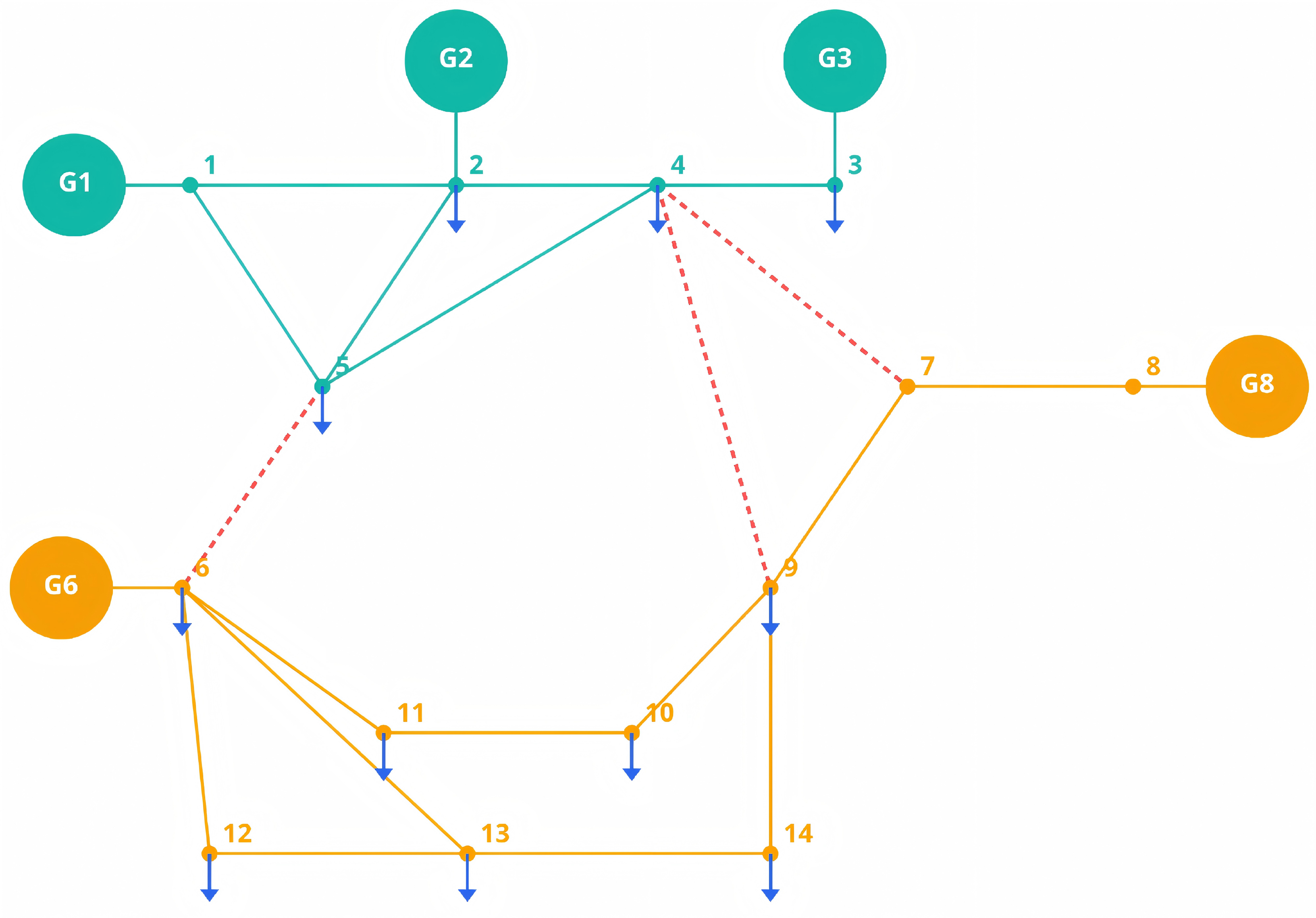} \\
\footnotesize (a) IEEE 9-bus & \footnotesize (b) IEEE 14-bus \\[4pt]
\includegraphics[width=0.42\textwidth]{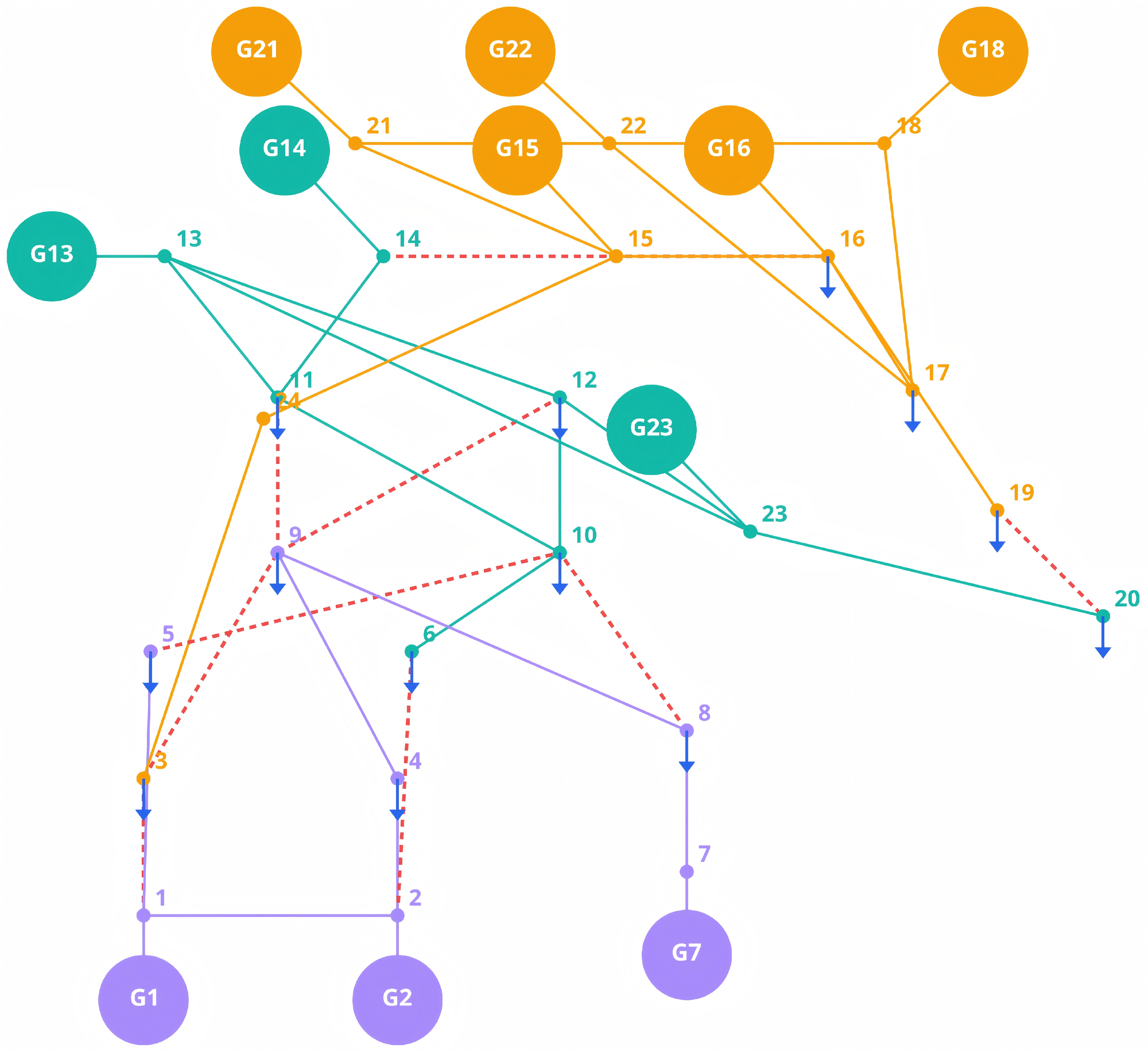} &
\includegraphics[width=0.42\textwidth]{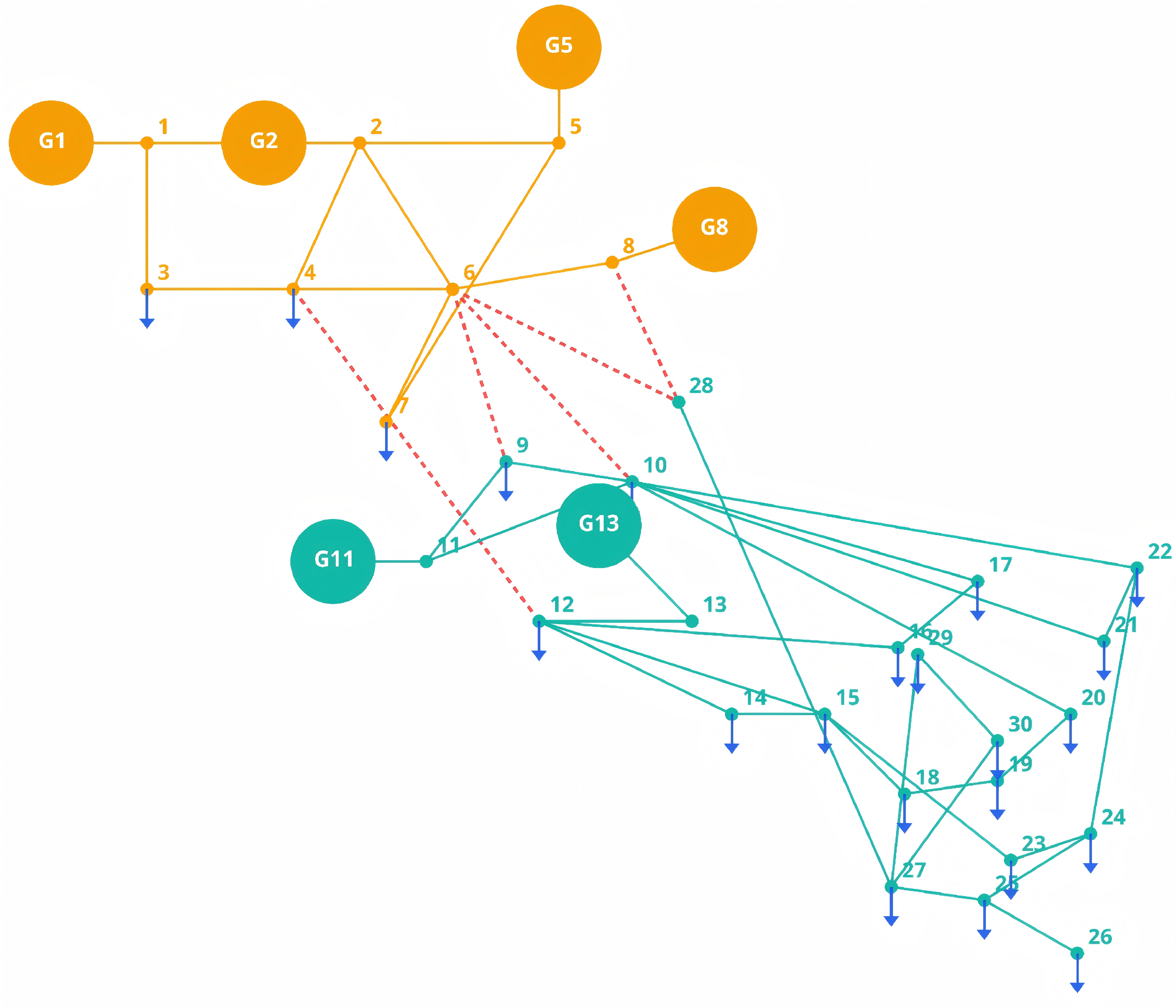} \\
\footnotesize (c) IEEE 24-bus & \footnotesize (d) IEEE 30-bus
\end{tabular}\par\vspace{4pt}
\begin{tabular}{@{}c@{\hspace{4pt}}c@{}}
\includegraphics[width=0.42\textwidth]{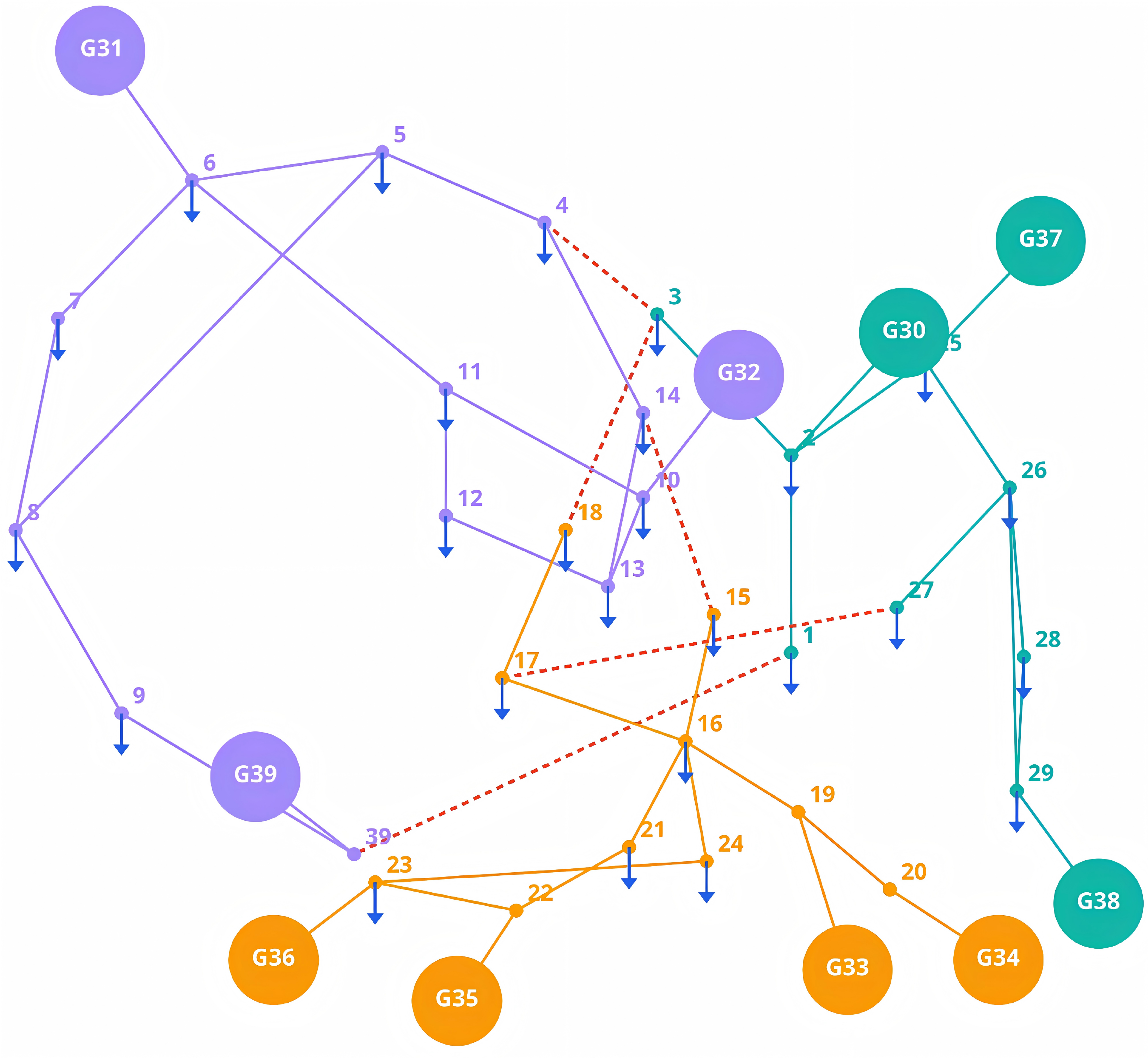} &
\includegraphics[width=0.42\textwidth]{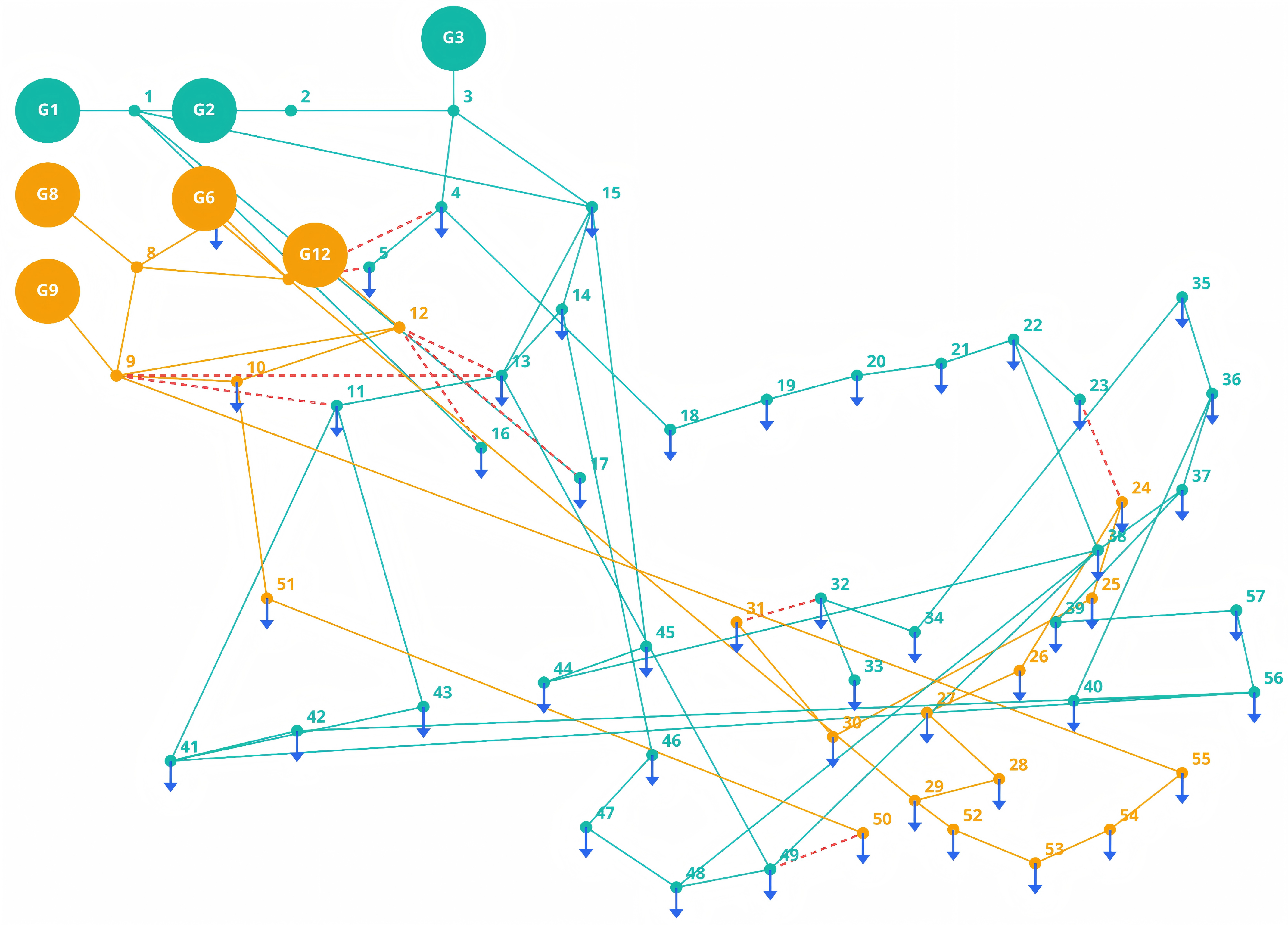} \\
\footnotesize (e) IEEE 39-bus & \footnotesize (f) IEEE 57-bus
\end{tabular}
\caption{Optimal islanding results generated by the proposed QAOA framework on the IEEE test systems. Bus colors indicate
the island assignment of each bus, and the red dashed lines denote the transmission interfaces removed to form the final
islands. The resulting partitions remain spatially compact and are separated by only a small number of
inter-area transmission cuts.}
\Description{Six IEEE benchmark network diagrams with colored bus groups indicating islands and red dashed lines indicating removed transmission interfaces.}
\label{fig:island_viz_ieee}
\end{figure}

\begin{figure}[!tbp]
\centering
\includegraphics[width=0.98\columnwidth,height=0.84\textheight,keepaspectratio]{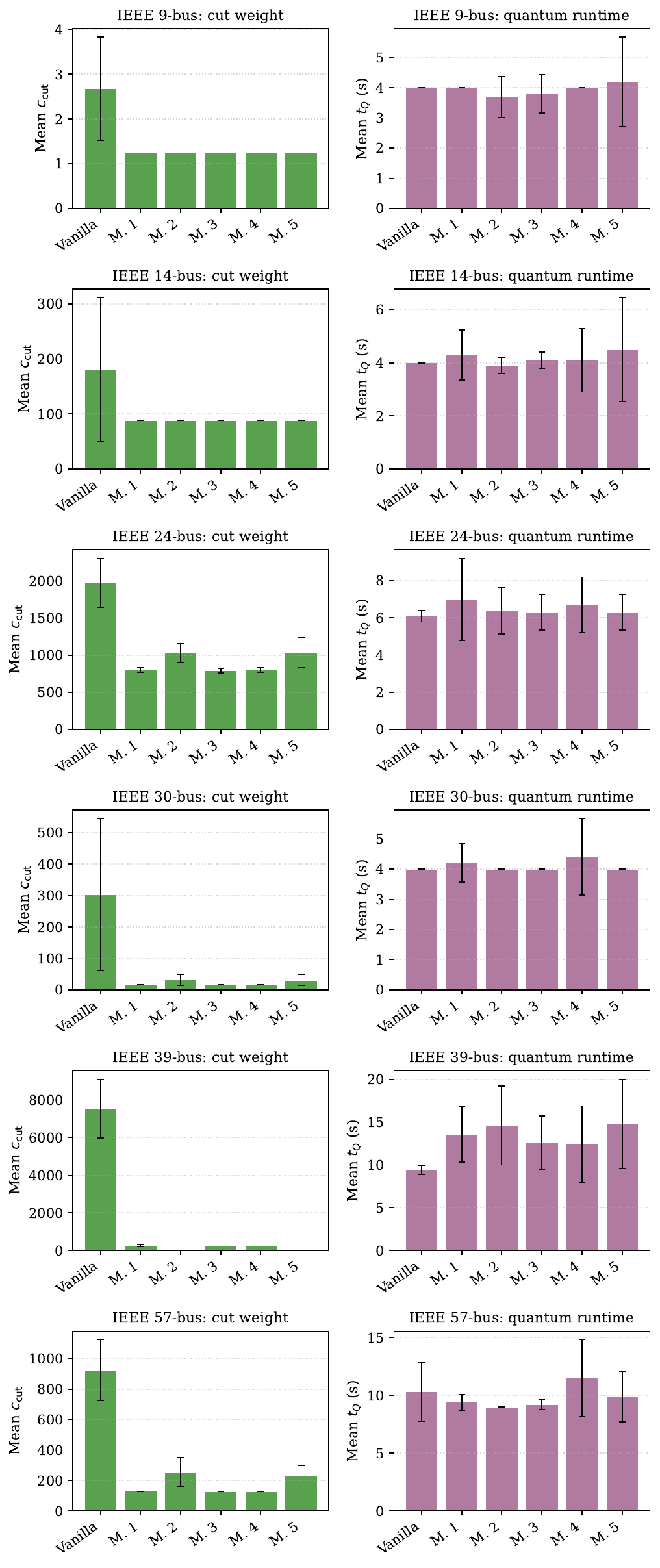}
\caption{Comparison of average cut quality and average quantum runtime across all the IEEE test systems for
vanilla QAOA and all the proposed post-processing methods. The left column shows the average cut values with sample
standard deviation bars, and the right column shows the average quantum runtime with sample standard
deviation bars. }
\Description{Bar charts comparing cut values and quantum runtimes for vanilla QAOA and five post-processing methods across six IEEE systems.}
\label{fig:postproc_cut_tq}
\end{figure}

\subsection{Benchmark Performance Comparison}

This subsection benchmarks the proposed hybrid quantum-classical islanding workflow across the six IEEE test
systems by comparing raw vanilla QAOA sampling with five post-processing methods denoted M.~1--M.~5 and with the classical Gurobi
optimum as a reference. The representative islanding solutions obtained from the QAOA-based workflow are listed in
Table~\ref{tab:islanding_partitions}, and Fig.~\ref{fig:island_viz_ieee} shows how these solutions map back to
the physical network structure through compact islands separated by a small set of cut interfaces.
The Table~\ref{tab:islanding_partitions} entries denote the buses assigned to each island.
Based on the coherency constraints, the IEEE~9-bus, 14-bus, 30-bus, and 57-bus systems are partitioned into
two islands, whereas the IEEE~24-bus and 39-bus systems are partitioned into three islands.
The IEEE~24-bus and 39-bus cases are therefore more challenging from the optimization perspective, since the
three-island requirement increases the combinatorial complexity of the partitioning problem and imposes
stricter demands on the subsequent connectivity and resource-balance repair steps.

Table~\ref{tab:postprocessing_bench}, Fig.~\ref{fig:postproc_cut_tq}, and
Fig.~\ref{fig:postproc_gap_gurobi} jointly summarize the benchmark outcomes in terms of cut quality and
quantum runtime. Without post-processing, vanilla QAOA yields mean cuts of $3.153$, $192.748$, $1768.581$,
$243.29$, $7544.796$, and $815.28$ on IEEE~9-, 14-, 24-, 30-, 39-, and 57-bus, respectively, against the
Gurobi optima of $1.238$, $88.241$, $776.206$, $16.863$, $228.993$, and $128.239$. Fig.~\ref{fig:postproc_cut_tq}
shows the same trend in absolute terms, while Fig.~\ref{fig:postproc_gap_gurobi} shows that the relative gap
to Gurobi grows from about $154\%$ on IEEE~9-bus to over $3000\%$ on IEEE~39-bus for vanilla QAOA.

Among the post-processing methods, M.~3 and~M.~4 provide the most consistent improvement. M.~4 reaches the
Gurobi optimum on every reported case, and M.~3 is exact on all but IEEE~24-bus, where its average performance is just slightly below M.~4.
Fig.~\ref{fig:postproc_cut_tq} shows that these cut improvements are accompanied by only a
limited increase in quantum runtime. For example, on IEEE~39-bus the mean~$t_Q$ rises from $9.4$~s for vanilla QAOA to
$12.6$~s and $12.4$~s for M.~3 and M.~4, respectively, while the mean cut drops from $7544.796$ to the
near-optimal range. On the same system, M.~2 and M.~5 still incur mean runtimes of $14.6$~s and $14.8$~s,
respectively, but fail to produce finite cut values, indicating that the additional runtime is not converted
into effective repair on this harder case. This is attributable to the algorithmic properties of M.~2 and M.~5 under the three-island structure: M.~2's Stage~2 accepts any connectivity-improving flip without a $Q$-descent condition, which can repeatedly raise the QUBO energy across repair moves and leave Stage~1 unable to recover a feasible configuration within the iteration budget; M.~5's large-$\mu$ unified descent prioritizes connectivity at every step, which on a three-island instance with more interdependent repair decisions can push the configuration away from the QUBO-feasible region. By contrast, M.~3 admits connectivity-repair flips only when they also strictly decrease $Q$, and M.~4 further relaxes the penalty weight when that condition blocks repair, allowing both methods to maintain QUBO feasibility throughout the descent.

The differences in solution quality across methods in Table~\ref{tab:postprocessing_bench} and
Figs.~\ref{fig:postproc_cut_tq}--\ref{fig:postproc_gap_gurobi} are attributed to the structural properties of
their respective repair formulations.
M.~1 adopts the simplest structure of a single-stage QUBO descent without connectivity repair,
which is sufficient for well-conditioned instances, as confirmed by the $^\star$ annotation on all six
systems. Its limitation is that connected partitions are obtained only when the QUBO penalty structure
is strong enough to guide the descent implicitly, a condition that becomes less reliable on larger
graphs, as reflected by the mean cut of $268.242\,(\pm 53.74)$ on IEEE~39-bus against the optimum
of $228.993$. M.~2 and~M.~5 fail on the largest instances
for opposite but complementary reasons: M.~2 applies an unconstrained Stage~2 that accepts any
connectivity-improving flip regardless of its effect on~$Q$, which can disrupt QUBO feasibility and
cause repeated oscillation between stages on complex instances; M.~5, despite being the only formulation that theoretically guarantees full feasibility
($z^* \in \mathcal{F}_Q \cap \mathcal{F}_C$) through a unified energy $\mathcal{Q} = Q + \mu C$,
underperforms in practice due to a fundamental tension in its design: the hierarchy condition
$\mu > \max_{z,j}|\Delta_j(z)|$ required to guarantee feasibility forces~$\mu$ to be large, which
biases every descent step toward connectivity repair regardless of the current QUBO state. Since QAOA
samples are typically near QUBO-feasible, this aggressive weighting steers the descent away from
the low-cut region highlighted in Table~\ref{tab:postprocessing_bench} and Fig.~\ref{fig:postproc_cut_tq}, and the algorithm may terminate at a
locally connected but cut-suboptimal solution. M.~3 resolves the instability of M.~2 by
admitting a Stage~2 connectivity flip only if it simultaneously decreases~$Q$, preserving strict
monotone descent throughout and achieving near-optimal solutions without an iteration bound. M.~4 extends this further by gradually relaxing the penalty weight $\lambda^{(k)}$ across outer
iterations, which unlocks Stage~2 flips that M.~3's strict descent condition would otherwise
reject, and thereby achieves the most consistent recovery of the Gurobi optimum across all systems
and trials, as seen most clearly in the near-zero gaps of Fig.~\ref{fig:postproc_gap_gurobi}. The penalty relaxation schedule provides a principled way to widen the search at each
iteration while preserving within-iteration strict descent, making M.~4 the most robust
post-processing variant among the five.
\newcommand{\PostprocBenchRows}{%
\multirow{2}{*}{\textbf{Vanilla}} & $\cutval$ & 3.1529 ($\pm$\,1.3254) & 192.748 ($\pm$\,87.51) & 1768.581 ($\pm$\,314.28) & 243.29 ($\pm$\,220.04) & 7544.796 ($\pm$\,1564.15) & 815.28 ($\pm$\,177.08) \\%
 & \textbf{$t_Q$} & 4 & 4 & 6.2 ($\pm$\,0.4) & 4 & 9.4 ($\pm$\,0.5) & 10 ($\pm$\,2.2) \\%
\midrule%
\multirow{2}{*}{\textbf{M.~1}} & $\cutval$ & 1.2378$^{\star}$ & 88.241$^{\star}$ & 785.887 ($\pm$\,21.65)$^{\star}$ & 16.863$^{\star}$ & 268.242 ($\pm$\,53.74)$^{\star}$ & 129.328 ($\pm$\,1)$^{\star}$ \\%
 & \textbf{$t_Q$} & 4 & 4 & 6.4 ($\pm$\,0.9) & 4 & 13.6 ($\pm$\,3.3) & 9.4 ($\pm$\,0.9) \\%
\midrule%
\multirow{2}{*}{\textbf{M.~2}} & $\cutval$ & 1.2378$^{\star}$ & 88.241$^{\star}$ & 1021.143 ($\pm$\,147.25) & 30.576 ($\pm$\,17.31)$^{\star}$ & {-}{-}{-} & 186.938 ($\pm$\,42.8) \\%
 & \textbf{$t_Q$} & 3.4 ($\pm$\,0.9) & 3.8 ($\pm$\,0.4) & 6 & 4 & 14.6 ($\pm$\,4.6) & 9 \\%
\midrule%
\multirow{2}{*}{\textbf{M.~3}} & $\cutval$ & 1.2378$^{\star}$ & 88.241$^{\star}$ & 784.44 ($\pm$\,18.41)$^{\star}$ & 16.863$^{\star}$ & 228.993$^{\star}$ & 128.239 ($\pm$\,0)$^{\star}$ \\%
 & \textbf{$t_Q$} & 3.6 ($\pm$\,0.9) & 4.2 ($\pm$\,0.4) & 6.6 ($\pm$\,1.3) & 4 & 12.6 ($\pm$\,3.1) & 9 \\%
\midrule%
\multirow{2}{*}{\textbf{M.~4}} & $\cutval$ & 1.2378$^{\star}$ & 88.241$^{\star}$ & 776.206$^{\star}$ & 16.863$^{\star}$ & 228.993$^{\star}$ & 128.239 ($\pm$\,0)$^{\star}$ \\%
 & \textbf{$t_Q$} & 4 & 4.2 ($\pm$\,1.8) & 6.8 ($\pm$\,1.8) & 4.8 ($\pm$\,1.8) & 12.4 ($\pm$\,4.5) & 11.6 ($\pm$\,3.2) \\%
\midrule%
\multirow{2}{*}{\textbf{M.~5}} & $\cutval$ & 1.2378$^{\star}$ & 88.241$^{\star}$ & 911.503 ($\pm$\,82.43) & 24.685 ($\pm$\,8.14)$^{\star}$ & {-}{-}{-} & 225.254 ($\pm$\,8.24) \\%
 & \textbf{$t_Q$} & 3.6 ($\pm$\,0.9) & 5 ($\pm$\,2.8) & 6 & 4 & 14.8 ($\pm$\,5.2) & 10.6 ($\pm$\,3) \\%
\midrule%
\textbf{Gurobi} & $\cutval$ & 1.2378 & 88.241 & 776.206 & 16.863 & 228.993 & 128.239 \\%
}%

\begin{table}[!t]
\centering
\footnotesize
\setlength{\tabcolsep}{2.5pt}
\renewcommand{\arraystretch}{1.05}
\caption{Cut quality and quantum runtime, reported in seconds, for vanilla QAOA, post-processing methods M.~1 through M.~5, and
Gurobi on the IEEE test systems
($^\star$ indicates that at least one trial attained the Gurobi-optimal cut value).}
\label{tab:postprocessing_bench}
\begin{tabular}{@{}ll*{6}{c}@{}}
\toprule
\textbf{Method} & \makecell{$\cutval$\\[1pt]$t_Q$ (s)}
& \textbf{9-bus} & \textbf{14-bus} & \textbf{24-bus} & \textbf{30-bus} & \textbf{39-bus} & \textbf{57-bus} \\
\midrule
\PostprocBenchRows
\bottomrule
\end{tabular}
\end{table}

\begin{table}[!tbp]
\centering
\scriptsize
\setlength{\tabcolsep}{4pt}
\renewcommand{\arraystretch}{1.0}
\caption{Approximate vanilla-QAOA depth and shot requirements for matching the solution quality attained by the
proposed hybrid workflow.}
\label{tab:vanilla_resource_req}
\begin{tabular}{@{}lcc@{}}
\toprule
\textbf{IEEE case} & \textbf{QAOA layers} & \textbf{Shots} \\
\midrule
9-bus & 9 & 10,000 \\
14-bus & 20 & 8,000,000 \\
24-bus & N/A & N/A \\
30-bus & N/A & N/A \\
39-bus & N/A & N/A \\
57-bus & N/A & N/A \\
\bottomrule
\end{tabular}
\end{table}

\begin{figure}[!tbp]
\centering
\includegraphics[width=\columnwidth]{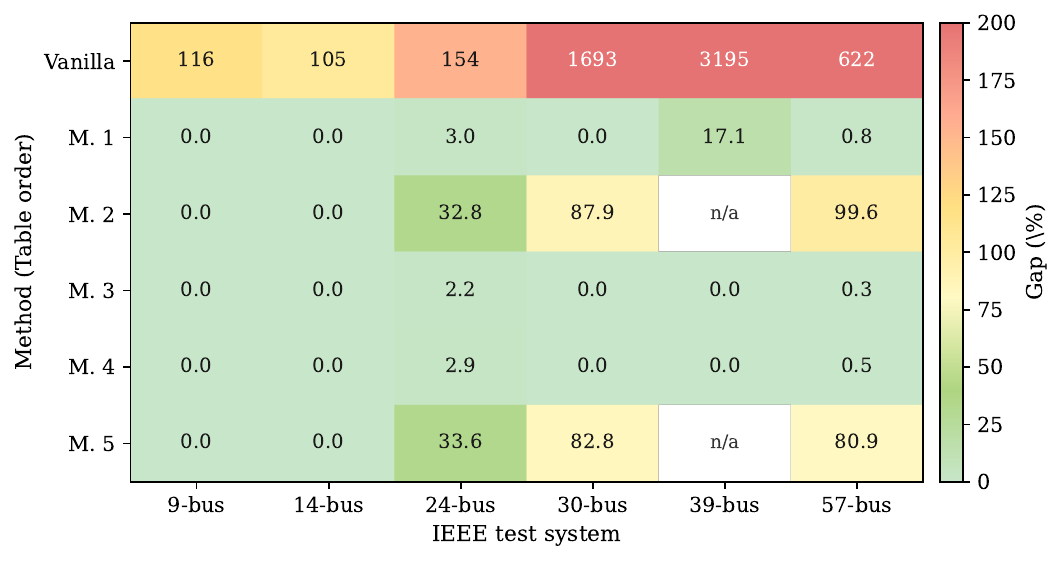}
\caption{Percent gap of the average cut value from the corresponding Gurobi benchmark across all the IEEE test
systems for vanilla QAOA and post-processing methods M.~1 through M.~5. Smaller values indicate closer agreement with the
classical optimum, and the cell labels report the corresponding percentage for each method and system.}
\Description{Heatmap of percent cut-value gaps from Gurobi for vanilla QAOA and five post-processing methods across six IEEE systems.}
\label{fig:postproc_gap_gurobi}
\end{figure}

Fig.~\ref{fig:postproc_feasibility} complements the cut-quality results in
Table~\ref{tab:postprocessing_bench} and Figs.~\ref{fig:postproc_cut_tq}--\ref{fig:postproc_gap_gurobi} by
showing how consistently each method concentrates probability mass on feasible partitions. Vanilla QAOA has
essentially zero feasible probability on all six systems, so its sampling distribution rarely places
meaningful weight on directly usable partitions. M.~1 improves this on the smaller cases, reaching $25.3\%$ on
IEEE~9-bus and $17.9\%$ on IEEE~30-bus, but it drops to only $1.1\%$ on IEEE~39-bus, which is consistent with
the cut degradation already seen in Table~\ref{tab:postprocessing_bench}. The strongest behavior again comes
from M.~3 and M.~4: even on the harder three-island cases they keep nontrivial feasible concentration, rising
to $13.1\%$ and $13.2\%$ on IEEE~39-bus, and on IEEE~57-bus M.~4 reaches $67.3\%$ compared with $44.1\%$ for
M.~1 and about $0.1\%$ for both M.~2 and M.~5. Even for IEEE~24-bus, where all methods are more challenged,
M.~4 still attains the highest mean feasible probability at $3.9\%$. This pattern reinforces the earlier
method comparison: repair is most effective when it restores connectivity without losing control of the
underlying QUBO descent, whereas unconstrained or overly connectivity-dominated repair either disperses or
misdirects the sampled probability mass.

\begin{figure}[!tbp]
\centering
\includegraphics[width=\columnwidth]{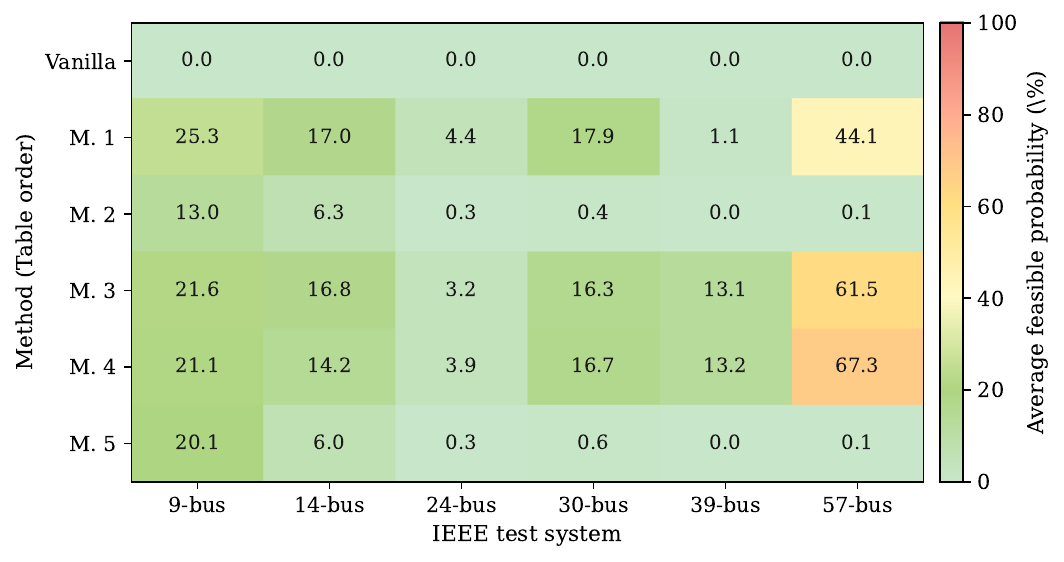}
\caption{Average feasible probability across all the IEEE test systems for vanilla QAOA and
post-processing methods M.~1 through M.~5, reported in percent. Larger values indicate that the sampling distribution places
more weight on feasible islanding partitions, and the cell labels give the corresponding average percentage for
each method and system.}
\Description{Heatmap of feasible-sample probabilities for vanilla QAOA and five post-processing methods across six IEEE systems.}
\label{fig:postproc_feasibility}
\end{figure}

Table~\ref{tab:vanilla_resource_req} further highlights the quantum-resource savings enabled by the proposed
post-processing strategy by reporting the approximate vanilla-QAOA circuit layers and shot count required to
match the solution quality achieved by the post-processed pipeline. On IEEE~9-bus, the proposed framework
reaches that quality with $R=2$ layers and $S=100$ shots, whereas vanilla QAOA requires $9$ layers and
$10{,}000$ shots, corresponding to reductions of $4.5\times$ in circuit depth and $100\times$ in sampling
cost. The disparity widens sharply on IEEE~14-bus, where $20$ layers and $8{,}000{,}000$ shots are required,
that is, a $10\times$ increase in depth and a $8000\times$ increase in shots relative to the post-processed
configuration ($R=2$, $S=1000$). For IEEE~24-bus and larger systems, the solution quality reached by the
post-processed pipeline is no longer matched by vanilla QAOA within a practically attainable circuit-layers and
shot budget (denoted as N/A in Table~\ref{tab:vanilla_resource_req}), underscoring the role of post-processing in extending the usable range of the quantum workflow.

\subsection{Noise Robustness Analysis}

In this subsection, the IBM cloud simulator is used under both an ideal noiseless setting and the \texttt{FakeMarrakesh} noise model to evaluate the noise robustness of the proposed framework. \texttt{FakeMarrakesh} replicates the calibrated noise profile of the 156-qubit ibm\_marrakesh processor, including per-gate depolarizing errors, T1/T2 thermal relaxation, and readout errors derived from live device data~\cite{ibm_processor_types_2026,ibm_fakemarrakesh_2026,ibm_compute_resources_2026}, and is therefore directly representative of the noise encountered during the hardware runs in Table~\ref{tab:postprocessing_bench}. The study is restricted to the IEEE~9-bus case, as its 32-qubit circuit (Table~\ref{tab:exp_setup}) is the only configuration that fits within the simulator's 32-qubit ceiling, while all other benchmark instances require strictly more qubits and therefore fall outside this constraint. The hardware results in Table~\ref{tab:postprocessing_bench} are used as the real-device reference, while Table~\ref{tab:noise_sim_comparison} contrasts the same instance under ideal and calibrated-noise simulation. Together with the feasible-probability evidence in Fig.~\ref{fig:postproc_feasibility}, these results provide a cross-backend assessment of whether the proposed post-processing framework maintains solution quality and feasible-sample concentration when moving from ideal simulation to calibrated simulator noise and real quantum hardware.

\begin{table}[!tbp]
\centering
\scriptsize
\setlength{\tabcolsep}{3pt}
\renewcommand{\arraystretch}{1.05}
\caption{Impact of noise on IEEE~9-bus simulation results ($^\star$: Gurobi-optimal cut found).}
\label{tab:noise_sim_comparison}
\begin{threeparttable}
\resizebox{\columnwidth}{!}{%
\begin{tabular}{@{}lcccc@{}}
\toprule
\textbf{Method} &
\makecell{\textbf{Ideal}\\$\cutval/\bar p_{\mathrm{feas}}$} &
\makecell{\textbf{Marrakesh}\\$\cutval/\bar p_{\mathrm{feas}}$} &
\makecell{\textbf{Ideal}\\$t_Q$} &
\makecell{\textbf{Marrakesh}\\$t_Q$} \\
\midrule
Vanilla & 2.002 ($\pm$\,1.215) / 0.000 ($\pm$\,0.000) & 2.101 ($\pm$\,1.391) / 0.000 ($\pm$\,0.000) & 0.03 ($\pm$\,0.01) & 2.98 ($\pm$\,0.09) \\
M.~1 & 1.238 ($\pm$\,0.000)$^\star$ / 0.266 ($\pm$\,0.033) & 1.238 ($\pm$\,0.000)$^\star$ / 0.250 ($\pm$\,0.038) & 0.02 ($\pm$\,0.00) & 2.93 ($\pm$\,0.13) \\
M.~2 & 1.238 ($\pm$\,0.000)$^\star$ / 0.136 ($\pm$\,0.040) & 1.238 ($\pm$\,0.000)$^\star$ / 0.136 ($\pm$\,0.050) & 0.02 ($\pm$\,0.00) & 3.06 ($\pm$\,0.08) \\
M.~3 & 1.238 ($\pm$\,0.000)$^\star$ / 0.222 ($\pm$\,0.062) & 1.238 ($\pm$\,0.000)$^\star$ / 0.220 ($\pm$\,0.067) & 0.02 ($\pm$\,0.00) & 2.87 ($\pm$\,0.12) \\
M.~4 & 1.238 ($\pm$\,0.000)$^\star$ / 0.188 ($\pm$\,0.031) & 1.238 ($\pm$\,0.000)$^\star$ / 0.180 ($\pm$\,0.039) & 0.02 ($\pm$\,0.00) & 3.02 ($\pm$\,0.05) \\
M.~5 & 1.238 ($\pm$\,0.000)$^\star$ / 0.222 ($\pm$\,0.062) & 1.238 ($\pm$\,0.000)$^\star$ / 0.220 ($\pm$\,0.067) & 0.02 ($\pm$\,0.00) & 3.01 ($\pm$\,0.04) \\
\bottomrule
\end{tabular}
}
\end{threeparttable}
\end{table}

Table~\ref{tab:postprocessing_bench} first provides the real-hardware reference on \texttt{ibm\_marrakesh}, and Table~\ref{tab:noise_sim_comparison} then compares the same IEEE~9-bus case under ideal statevector simulation and the calibrated \texttt{FakeMarrakesh} noise model. These results highlight the noise robustness gained by the proposed post-processing framework. In terms of solution quality, vanilla QAOA has a mean cut value of $2.002$ in ideal simulation, $2.101$ under the noise model, and $3.153$ on real hardware, whereas M.~1--M.~5 retain the Gurobi-optimal cut value of $1.238$ across the ideal simulator, noisy simulator, and hardware settings. Consistent with Fig.~\ref{fig:postproc_feasibility}, the feasible probability results in Table~\ref{tab:noise_sim_comparison} show that vanilla QAOA remains at zero in both simulator columns, while the post-processed methods maintain similar values between the ideal and noisy simulator results, indicating that calibrated noise has little impact on feasible-sample concentration for this instance. In terms of backend execution time, the \texttt{FakeMarrakesh} simulator requires about $2.87$--$3.06$~s, and the real-hardware $t_Q$ values in Table~\ref{tab:postprocessing_bench} are of the same order for the IEEE~9-bus case. Thus, the numerical evidence indicates that post-processing preserves solution quality and feasible-sample concentration under calibrated noise and real hardware execution without introducing a large backend-time penalty for this case. More broadly, the real-hardware results in Table~\ref{tab:postprocessing_bench} show that the proposed workflow can recover Gurobi-optimal solutions across all six benchmark systems with shallow QAOA circuits and limited sampling resources. This supports the effectiveness of the proposed hybrid framework as a noise-resilient quantum optimization workflow, in which shallow QAOA sampling provides useful candidate structure and classical post-processing converts that structure into high-quality, physically feasible islanding solutions.

\begin{figure}[!tbp]
\centering
\includegraphics[width=\columnwidth,height=0.72\textheight,keepaspectratio]{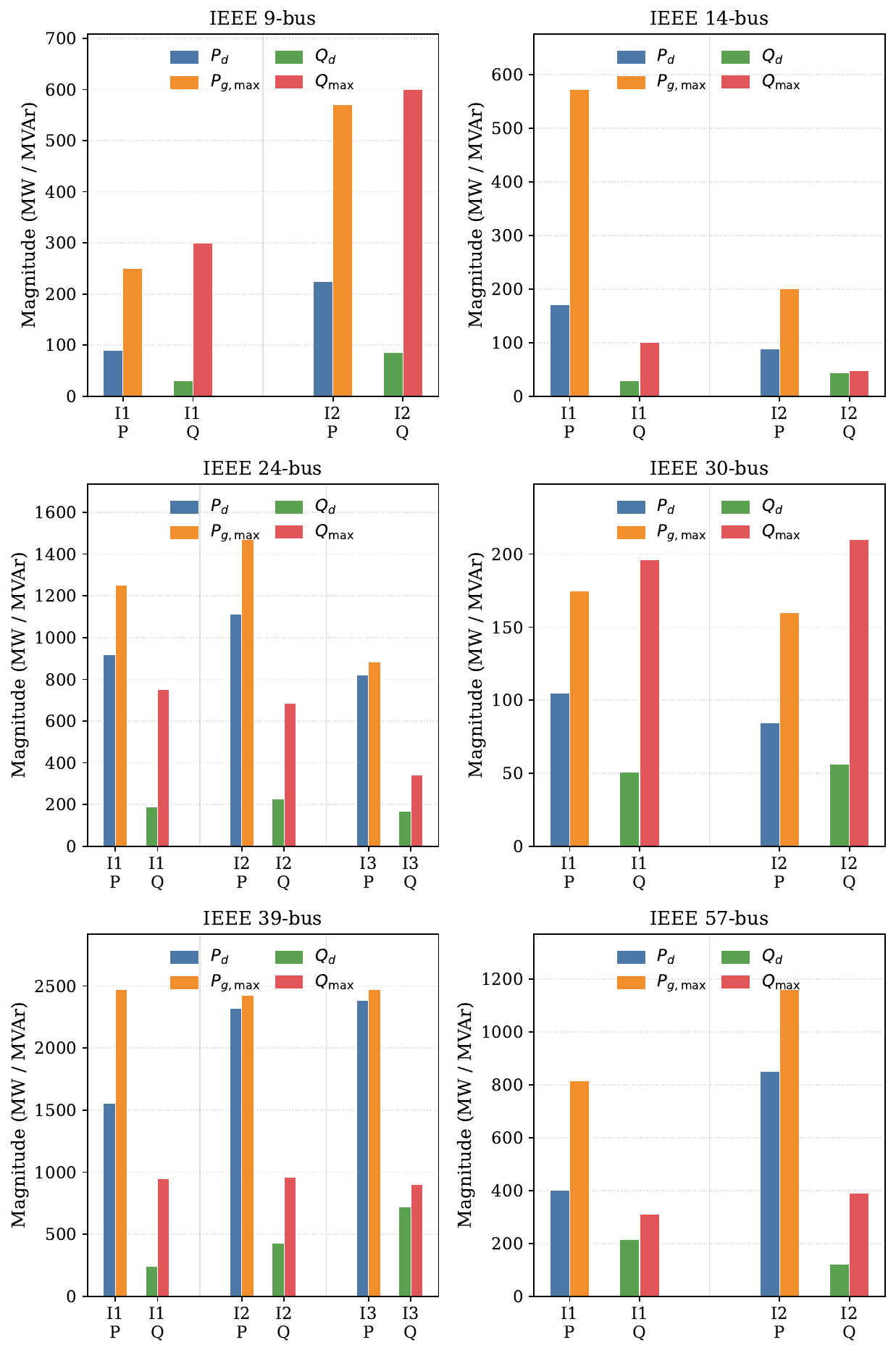}
\caption{Island-level comparison of load and generation capability for the islanding solutions obtained by the
proposed QAOA framework. For each case, the grouped bars compare active-power demand with maximum
active-generation capacity and reactive-power demand with reactive upper capability.}
\Description{Grouped bar charts comparing demand and generation capability for active and reactive power in each island across the benchmark systems.}
\label{fig:island_power_balance}
\end{figure}

\subsection{Island Power-Balance Validation}

This subsection applies the two post-pipeline physical validation steps introduced in Section~\ref{formulation} to the representative islanding solutions selected by the proposed framework. The power flow check verifies whether each island admits a consistent post-separation active-power transfer pattern, and the power support check confirms that each island satisfies $P_{g,\max}>P_d$ and $Q_{\max}>Q_d$, ensuring self-sufficiency in both active and reactive power after separation.

All islands in the selected partitions pass the power flow check, indicating
that each isolated subsystem admits a consistent post-separation operating solution. Fig.~\ref{fig:island_power_balance}
evaluates the active- and reactive-power support of the same partitions by comparing, for every island, the
active-power demand with the available active-generation capacity and the reactive-power demand with the
reactive upper capability.
Across all six IEEE systems, every island in the selected partition satisfies $P_{g,\max}>P_d$ and
$Q_{\max}>Q_d$, which indicates that the post-processed solutions are not only low-cut graph partitions but
also self-supporting candidates after separation.

The most constrained cases in Fig.~\ref{fig:island_power_balance} are still feasible but exhibit visibly reduced operating margins. In the IEEE~24-bus system, Island~3 has $P_d=821~\mathrm{MW}$ and $P_{g,\max}=884~\mathrm{MW}$, so the
active-power margin is only $63~\mathrm{MW}$, whereas the other two islands retain margins above
$300~\mathrm{MW}$. In the IEEE~14-bus system, Island~2 is the tightest reactive case, with
$Q_d=44.1~\mathrm{MVAr}$ against $Q_{\max}=48~\mathrm{MVAr}$, leaving only $3.9~\mathrm{MVAr}$ of upper reactive
headroom. By contrast, the IEEE~39-bus and IEEE~57-bus partitions remain comfortably feasible overall; for
example, the three IEEE~39-bus islands preserve active-power margins of $915.9$, $109.9$, and
$86.97~\mathrm{MW}$, while IEEE~57-bus retains $414.48~\mathrm{MW}$ and $310.6~\mathrm{MW}$ of active margin in its
two islands. Fig.~\ref{fig:island_power_balance} therefore shows that the selected post-processed solutions preserve both the intended
partition structure and the basic island-level power-support capability needed for practical implementation.

\section{Conclusion}
\label{conclusion}

This paper presents a hybrid quantum-classical framework for controlled power-system islanding and pioneers
the application of QAOA for power-system islanding with explicit feasibility-critical constraints. The proposed
framework combines quantum optimization with a structured post-processing design that converts sampled QAOA
outputs into feasible islanding decisions when important physical requirements, such as exact connectivity and
operational consistency, cannot be enforced by the QUBO model alone, thereby extending quantum optimization
from unconstrained graph partitioning to physically meaningful islanding decisions.

The numerical studies on the IEEE 9-, 14-, 24-, 30-, 39-, and 57-bus systems show that this hybrid design is
necessary for obtaining high-quality solutions from shallow QAOA sampling. Among all the proposed post-processing
methods, M.~3 and especially M.~4 deliver the strongest overall performance, with M.~4 recovering the
Gurobi-optimal cut on all reported benchmark cases while requiring only modest additional runtime. The
benchmark comparison further shows that the proposed workflow attains solution quality that would require much
deeper and more heavily sampled vanilla QAOA on the small systems and is not matched by vanilla QAOA within
the current practical resource budget on the larger systems. The noise robustness analysis further indicates that
the post-processing framework preserves high-quality feasible solutions under realistic quantum noise, supporting
the resilience of the proposed near-term quantum workflow. The selected partitions also pass the
power flow validation and retain adequate active- and reactive-power support margins, indicating that the
resulting islands are physically meaningful candidates for post-disturbance operation.

Future work will focus on improving the quantum efficiency and practical scalability of the framework through
power-system-aware ansatz design, structure-informed mixers and initialization strategies, and hardware-aware
implementations that further reduce sampling cost while preserving feasibility and solution quality.

\clearpage
\begin{acks}
This work is supported by the Office of Naval Research under the award N00014-22-1-2504, the National Science Foundation under the awards OAC-2417773 and ECCS-2413237/2413238.
\end{acks}
\bibliographystyle{ACM-Reference-Format}
\bibliography{References}
\appendix
\section{Formal Analysis of Post-Processing Methods}
\label{app:postprocessing}

This appendix provides rigorous proofs of the descent, termination, and feasibility properties
of all the proposed post-processing methods.  M.~1 and M.~2 are analyzed first to identify precisely
where their guarantees break down; M.~3 through M.~5 are then proved against those baselines.

\subsection*{Standing Notation and Conditions}

Throughout the appendix, $z \in \{0,1\}^N$ denotes an island-assignment configuration,
$e_j$ the unit flip vector at coordinate $j$, and $\Delta^E_j(z) := E(z \oplus e_j) - E(z)$
the per-flip change of any energy $E$.  For $E = Q$ we write $\Delta_j(z)$;
for $E = C$ we write $\Delta^C_j(z)$;
for $E = \mathcal{Q}$ we write $\Delta^u_j(z)$;
and for $E = Q^{(r)}$ we write $\Delta^{(r)}_j(z)$.
Define
\begin{equation}
  \Delta_{\max} \;:=\; \max_{z,\,j}\,\bigl|\Delta^{H_{\mathrm{cut}}}_j(z)\bigr|
  \label{eq:deltamax}
\end{equation}
as the maximum single-flip change in the cut objective.  Define also
$\mathcal{F}^* := \mathcal{F}_Q \cap \mathcal{F}_C$.

Each constraint penalty takes either the hinge-loss form used in the main text
(Section~\ref{sec:post_common}),
\[
  H^m(z) \;=\; \max\!\bigl(0,\; S^m(z) - b^m_{\max},\; b^m_{\min} - S^m(z)\bigr),
\]
or, after inequality constraints are converted to equalities by binary slack variables, the
equivalent squared form $H^m(z) = \bigl(S^m(z) - b^m\bigr)^2$ with $b^m := b^m_{\min} = b^m_{\max}$.
Both forms satisfy $H^m(z)=0$ iff the constraint holds, and the coefficients $a^m_j$ are
integer-valued ($a^m_j\in\mathbb{Z}$), so $a^m_j\neq 0$ implies $|a^m_j|\geq 1$.

The following conditions are assumed where indicated:
\begin{description}
  \item[(A1)] $\lambda_m > \Delta_{\max}$ for every $m \in \{1,\ldots,5\}$.
  \item[(A2)] For every $z$ violating a linear constraint (one-hot, min-size, generation,
    or load), there exist a violated constraint~$m$ and a coordinate $j \in \mathcal{V}^m$
    with $|a^m_j| = 1$ and $(1 - 2z_j)\,a^m_j > 0$.  The coherency penalty is quadratic
    and is not covered by this condition.
  \item[(A3)] Every QUBO-infeasible configuration $z \notin \mathcal{F}_Q$ admits at least one
    single-variable flip $j$ with $\Delta_j(z) < 0$.
  \item[(A4)] For every QUBO-infeasible $z$ with $C(z) = 0$, there exists a flip $j$ with
    $\Delta_j(z) < 0$ and $\Delta^C_j(z) = 0$.
\end{description}

\begin{remark}
In~(A2), the coefficient $a^m_j$ is the multiplier of variable $z_j$ in the constraint sum
$S^m(z) = \sum_{i \in \mathcal{V}^m} a_i^m z_i$.  For example, in a one-hot constraint
$\sum_k y_{i,k} = 1$, each assignment variable $y_{i,k}$ has coefficient $+1$.
In the islanding QUBO of Section~\ref{quantum}, assignment variables participate in multiple
penalties simultaneously (one-hot, size, generation/load, and coherency), so the
disjoint variable-set property $\mathcal{V}^m \cap \mathcal{V}^n = \emptyset$ does not
hold globally and is not assumed in any theorem below.  Condition~(A3), namely that every QUBO-infeasible configuration
admits at least one improving single-variable flip, is therefore used directly as the
hypothesis of Theorem~\ref{thm:m3_qubo}.  For some partition-constraint violations
(min-size, generation, or load), an exclusive-slack repair argument may identify such a flip,
because a slack variable belongs only to its own penalty term.  However, this argument does
not apply uniformly across all infeasibility patterns, and in particular not to one-hot or
coherency violations, where assignment variables participate in several penalties
simultaneously.  For that reason, we retain~(A3) as a stand-alone hypothesis rather than
deriving it globally from~(A1)--(A2).
The framework of~\cite{shirai2024postprocessing} proves feasibility guarantees under an
``independent constraints'' assumption (disjoint variable sets, $\mathcal{V}^m \cap \mathcal{V}^n = \emptyset$);
because the islanding QUBO has dependent constraints, their Theorem~2 does not apply
directly, and~(A3) is stated here as an explicit regularity hypothesis in its place.
\end{remark}

\begin{remark}[Scope of the theoretical guarantees]
Conditions~(A3), (A4), and the flip-existence condition of Lemma~\ref{lem:m5_conn} are
stated as explicit regularity hypotheses rather than derived properties, reflecting the
coupled-constraint structure of the islanding QUBO in which the same variable
participates in multiple penalty terms simultaneously.  This is a known trade-off when
extending post-processing theory beyond the independent-constraint setting
of~\cite{shirai2024postprocessing}: stronger structural assumptions enable cleaner
derivations but may not hold universally across all problem instances.  The proofs in
this appendix are fully rigorous under these hypotheses, and the hypotheses themselves
are mild regularity conditions that hold for the class of power network topologies
considered here.  Across all six IEEE benchmark systems tested in
Section~\ref{sec:numerical} (9- to 57-bus), no violation of any condition is observed:
every tested instance satisfies the regularity requirements, and the post-processing
methods consistently recover near-optimal, fully feasible islanding solutions, confirming
that the empirical performance substantially compensates for the theoretical generality
that the coupled-constraint setting necessarily trades away.
\end{remark}

An additional condition, denoted~(A5), is stated when needed for M.~5.

The following identity underlies every proof.

\begin{lemma}[Fundamental identity]\label{lem:fund}
For any energy function $E$, any configuration $z$, and any variable index $j$,
\[
  E(z \oplus e_j) = E(z) + \Delta^E_j(z).
\]
\end{lemma}

\begin{proof}
By definition,
\begin{equation}
  \Delta^E_j(z) = E(z \oplus e_j) - E(z),
\end{equation}
which rearranges to $E(z \oplus e_j) = E(z) + \Delta^E_j(z)$.
\end{proof}

\subsection{M.~1: QUBO-Only Descent}

M.~1 runs the Stage~1 greedy rule on $Q$ alone and performs no
connectivity repair.

\begin{theorem}[Strict descent and QUBO feasibility]\label{thm:m1_descent}
Under~(A3), every flip executed by M.~1 strictly decreases $Q$,
and the output $z^*$ belongs to $\mathcal{F}_Q$.
\end{theorem}

\begin{proof}
\textit{Strict descent:}
Every flip is executed only when $\Delta_{j^*}(z^{(t)}) < 0$.
By Lemma~\ref{lem:fund},
\begin{equation}
  Q(z^{(t+1)})
  = Q(z^{(t)}) + \underbrace{\Delta_{j^*}(z^{(t)})}_{<\,0}
  < Q(z^{(t)}).
\end{equation}

\textit{Termination:}
Suppose configuration $z$ is visited at times $t_1 < t_2$.
Chaining strict descent along the $t_2 - t_1$ intervening steps yields
\begin{equation}
  Q\bigl(z^{(t_2)}\bigr) < Q\bigl(z^{(t_1)}\bigr).
  \label{eq:m1_term_strict}
\end{equation}
But $z^{(t_2)} = z^{(t_1)}$ implies
\begin{equation}
  Q\bigl(z^{(t_2)}\bigr) = Q\bigl(z^{(t_1)}\bigr),
\end{equation}
contradicting~\eqref{eq:m1_term_strict}.
Hence all visited configurations are distinct, and since
\begin{equation}
  \bigl|\{0,1\}^N\bigr| = 2^N < \infty,
\end{equation}
the algorithm terminates after finitely many flips.

\textit{QUBO feasibility:}
At termination~$z^*$, the greedy rule finds no improving flip, so
\begin{equation}
  \Delta_j(z^*) \;\geq\; 0 \quad \forall\, j.
  \label{eq:m1_local_min}
\end{equation}
By Theorem~\ref{thm:m3_qubo} (proved below independently of this result), every QUBO-infeasible configuration
possesses some $j$ with $\Delta_j(z) < 0$, contradicting~(\ref{eq:m1_local_min}).
Therefore $z^* \in \mathcal{F}_Q$.
\end{proof}

\begin{theorem}[DFS feasibility is not guaranteed]\label{thm:m1_nodfs}
M.~1 does not guarantee $z^* \in \mathcal{F}_C$.
\end{theorem}

\begin{proof}
Define the connectivity violation count
\begin{equation}
  C(z) \;:=\; \sum_{k=1}^{K}\bigl(c(G[V_k(z)]) - 1\bigr) \;\geq\; 0,
  \label{eq:C_def}
\end{equation}
where $c(H)$ denotes the number of connected components of graph $H$.
Since $Q(z) = H_{\mathrm{cut}}(z) + \sum_{m}\lambda_m H^m(z)$ does not contain
$C(z)$, the single-flip change satisfies
\begin{equation}
  \Delta_j(z)
  = \Delta^{H_{\mathrm{cut}}}_j(z)
    + \sum_{m} \lambda_m\,\Delta^m_j(z),
  \label{eq:m1_delta_decomp}
\end{equation}
independently of $\Delta^C_j(z)$.
Thus a connectivity-improving flip ($\Delta^C_j(z) < 0$) may still satisfy
\begin{equation}
  \Delta_j(z) = \Delta^{H_{\mathrm{cut}}}_j(z) + \sum_{m}\lambda_m\,\Delta^m_j(z) \geq 0,
\end{equation}
and be rejected, whereas a connectivity-worsening flip ($\Delta^C_j(z) > 0$) may satisfy
\begin{equation}
  \Delta_j(z) < 0
\end{equation}
and be accepted.

To exhibit a DFS-infeasible local minimum, consider any network instance
in which the QUBO-optimal partition $\hat{z}$ has at least one island
whose induced subgraph is disconnected.  Starting from any $z^{(0)}$ in
the basin of attraction of $\hat{z}$ under the greedy $Q$-descent,
M.~1 converges to $z^* = \hat{z}$, for which
\begin{equation}
  \Delta_j(z^*) \geq 0\ \ \forall\, j
  \quad\text{(by~\eqref{eq:m1_local_min})},
  \qquad
  C(z^*) > 0.
\end{equation}
Hence $z^* \notin \mathcal{F}_C$.
\end{proof}

\subsection{M.~2: Unconstrained Two-Stage Repair}

M.~2 appends a Stage~2 connectivity repair to M.~1. Stage~2 picks the candidate
in $\mathcal{J}(z)$ with the smallest $\Delta_j(z)$ and flips it
unconditionally, without checking the sign of $\Delta_{j^*}$.

\begin{theorem}[Stage~2 does not guarantee descent]\label{thm:m2_nodescent}
Under~(A1) and~(A3), a Stage~2 flip can strictly increase $Q$.
\end{theorem}

\begin{proof}
Let $z$ be the Stage~1 output, so $z \in \mathcal{F}_Q$ and
$C(z) > 0$.  Construct an instance in which the set of connectivity-improving
flips is a singleton: $\mathcal{J}(z) = \{j^*\}$.  Specifically, let $j^*$
reassign a bus from island~$k'$ to a neighboring island, merging two disconnected
components but reducing island~$k'$ below the minimum required size~$N_{\min}$, so that
the min-size constraint~$m$ is violated.  Because $\mathcal{J}(z)=\{j^*\}$,
Stage~2 must select~$j^*$.  Then
\begin{equation}
  H^m(z) = 0,
  \qquad
  H^m(z \oplus e_{j^*}) \geq 1,
  \qquad
  \Delta^m_{j^*}(z) \geq 1.
\end{equation}
Since $z \in \mathcal{F}_Q$ every penalty is zero, so flipping $j^*$ cannot decrease any
penalty term: $\Delta^{m'}_{j^*}(z) \geq 0$ for all $m'$.  Therefore, using~(A1),
\begin{equation}
  \Delta_{j^*}(z)
  = \Delta^{H_{\mathrm{cut}}}_{j^*}(z)
    + \sum_{m'}\lambda_{m'}\Delta^{m'}_{j^*}(z)
  \geq -\Delta_{\max} + \lambda_m\cdot 1
  > 0.
\end{equation}
Stage~2 executes $j^*$, so by Lemma~\ref{lem:fund},
\begin{equation}
  Q(z \oplus e_{j^*}) = Q(z) + \Delta_{j^*}(z) > Q(z).
\end{equation}
\end{proof}

\begin{theorem}[$Q$ need not be globally monotone]\label{thm:m2_nonmono}
Under~(A1) and~(A3), there exist instances for which the sequence
$Q(z^{(0)}), Q(z^{(1)}), \ldots$ produced by M.~2 is not monotonically non-increasing.
\end{theorem}

\begin{proof}
Let $z^{(0)}, z^{(1)}, \ldots, z^{(T)}$ denote the Stage~1 sub-sequence.
By Theorem~\ref{thm:m1_descent},
\begin{equation}
  Q\bigl(z^{(0)}\bigr) > Q\bigl(z^{(1)}\bigr) > \cdots > Q\bigl(z^{(T)}\bigr).
\end{equation}
At $z^{(T)}$ Stage~1 has no improving flip, so $z^{(T)} \in \mathcal{F}_Q$.
Suppose $C(z^{(T)}) > 0$.  Stage~2 selects $j^* \in \mathcal{J}(z^{(T)})$
unconditionally and sets
\begin{equation}
  z^{(T+1)} = z^{(T)} \oplus e_{j^*}.
\end{equation}
By Theorem~\ref{thm:m2_nodescent} there exist instances with $\Delta_{j^*}(z^{(T)}) > 0$, hence
\begin{equation}
  Q\bigl(z^{(T+1)}\bigr)
  = Q\bigl(z^{(T)}\bigr) + \Delta_{j^*}\bigl(z^{(T)}\bigr)
  > Q\bigl(z^{(T)}\bigr).
\end{equation}
Thus $Q(z^{(T)}) < Q(z^{(T+1)})$, contradicting monotone non-increase of~$(Q(z^{(t)}))_t$.
\end{proof}

\begin{theorem}[Nontermination without a finite outer bound~$K$]\label{thm:m2_cycle}
Under~(A1) and~(A3), without an outer iteration bound~$K$, M.~2 may cycle indefinitely.
\end{theorem}

\begin{proof}
In the constructed instance, $Q$ need not be monotone
(Theorem~\ref{thm:m2_nonmono}), so revisiting a configuration is not excluded
by an energy argument.  A concrete cycle arises as follows.
Let $z_A$ be the unique local minimum of $Q$ reachable from $z^{(0)}$,
with $C(z_A) > 0$.  Stage~2 flips $z_A$ to $z_B$, so that, as in Theorem~\ref{thm:m2_nodescent},
\begin{equation}
  Q(z_B) > Q(z_A).
  \label{eq:m2_cycle_stage2_increase}
\end{equation}
Stage~1 then descends $Q$ from $z_B$; if the basin of attraction of $z_A$ under Stage~1 contains
$z_B$, Stage~1 returns to~$z_A$.  Stage~2 then repeats the same flip to $z_B$, and the configuration
trajectory forms the periodic pattern
\begin{equation}
  z_A \xrightarrow{\text{Stage~2}} z_B \xrightarrow{\text{Stage~1}} z_A \longrightarrow \cdots,
  \label{eq:m2_cycle_diagram}
\end{equation}
which repeats indefinitely.  An outer bound $K$ is therefore necessary to guarantee termination.
\end{proof}

\subsection{M.~3: Constrained Two-Stage Descent}

\begin{theorem}[Global strict descent]\label{thm:m3_descent}
Every flip executed by M.~3, in either Stage~1 or Stage~2, strictly decreases
$Q$.
\end{theorem}

\begin{proof}
In Stage~1 a flip is executed only when $\Delta_{j^*}(z^{(t)}) < 0$.  By
Lemma~\ref{lem:fund},
\begin{equation}
  Q(z^{(t+1)})
  = Q(z^{(t)}) + \underbrace{\Delta_{j^*}(z^{(t)})}_{<\,0}
  < Q(z^{(t)}).
\end{equation}
In Stage~2 a flip is executed only when
\begin{equation}
  j^* \in \mathcal{J}(z),
  \qquad
  \Delta_{j^*}(z) < 0
\end{equation}
(see~(\ref{eq:m3_stage2})).  Lemma~\ref{lem:fund} then yields the same strict inequality with
$(z^{(t)},z^{(t+1)})$ replaced by the Stage~2 pair.  Every executed flip falls into one of these two
cases, so $Q$ strictly decreases at every step.
\end{proof}

\begin{theorem}[No cycles and finite termination]\label{thm:m3_term}
M.~3 visits no configuration twice and terminates in at most $2^N - 1$ flips,
with no outer iteration bound required.
\end{theorem}

\begin{proof}
Suppose, toward a contradiction, that some configuration is visited twice:
\begin{equation}
  z^{(t_1)} = z^{(t_2)},
  \qquad
  t_1 < t_2.
\end{equation}
By Theorem~\ref{thm:m3_descent}, every executed flip strictly decreases~$Q$.  Chaining these strict
inequalities from time~$t_1$ to time~$t_2$ yields
\begin{equation}
  Q\bigl(z^{(t_1)}\bigr)
  > Q\bigl(z^{(t_1+1)}\bigr)
  > \cdots
  > Q\bigl(z^{(t_2)}\bigr).
  \label{eq:m3_term_Q_chain}
\end{equation}
In particular, $Q(z^{(t_2)}) < Q(z^{(t_1)})$.  Yet \eqref{eq:m3_term_Q_chain} contradicts the identity
\begin{equation}
  z^{(t_1)} = z^{(t_2)}
  \quad\Longrightarrow\quad
  Q\bigl(z^{(t_1)}\bigr) = Q\bigl(z^{(t_2)}\bigr),
\end{equation}
so no configuration is visited twice.

The hypercube has finite cardinality
\begin{equation}
  \bigl|\{0,1\}^N\bigr| = 2^N,
\end{equation}
whence any trajectory with pairwise-distinct visited states has length at most $2^N - 1$ flips, and the
algorithm terminates.
\end{proof}

\begin{theorem}[Conditional QUBO feasibility at Stage~1 termination]\label{thm:m3_qubo}
Under condition~(A3), every Stage~1 local minimum belongs to
$\mathcal{F}_Q$.
\end{theorem}

\begin{proof}
We argue by contrapositive: it suffices to show that every $z \notin \mathcal{F}_Q$ admits
some flip~$j$ with $\Delta_j(z) < 0$, hence $z$ cannot be a Stage~1 local minimum.
Condition~(A3) gives exactly this implication: whenever $z \notin \mathcal{F}_Q$, there exists a
single-variable flip $j$ such that $\Delta_j(z) < 0$.  Therefore no QUBO-infeasible configuration
can satisfy the Stage~1 local-minimum condition $\Delta_j(z) \ge 0$ for all $j$.  Hence every
Stage~1 local minimum belongs to $\mathcal{F}_Q$.
\end{proof}

\begin{theorem}[Stage~2 failure under tight equality constraints]\label{thm:m3_fail}
Assume (A1), (A3), and integer coefficients $a^m_j\in\mathbb{Z}$.
Suppose Stage~1 has terminated at a configuration~$z$, so that $z\in\mathcal{F}_Q$ (by
Theorem~\ref{thm:m3_qubo}), and suppose at least one constraint $m$ is tight
(i.e.\ $S^m(z) = b^m$, equivalently $H^m(z)=0$, with $\mathcal{V}^m$ the set of variable
indices participating in constraint~$m$).
Then $\Delta_j(z) > 0$ for every $j \in \mathcal{V}^m$, so no
connectivity-improving flip in $\mathcal{V}^m$ can pass the Stage~2 descent gate.
\end{theorem}

\begin{proof}
\textit{Step 1 (all penalties are already zero).}
By Theorem~\ref{thm:m3_qubo}, Stage~1 termination together with~(A3) gives $z\in\mathcal{F}_Q$,
so $H^n(z) = 0$ for every constraint~$n$.  Because $H^n(z)\ge 0$ everywhere, it follows that
\begin{equation}
  \Delta^n_j(z) \;=\; H^n(z\oplus e_j) \;\geq\; 0
  \qquad \forall\,n,\ \forall\,j.
  \label{eq:m3_fail_all_nonneg}
\end{equation}

\textit{Step 2 (the targeted penalty strictly increases).}
Fix $j\in\mathcal{V}^m$.  Since $z_j\in\{0,1\}$ and $a^m_j\neq 0$,
\begin{equation}
  S^m(z \oplus e_j) \;=\; b^m + (1-2z_j)a^m_j \;\neq\; b^m.
  \label{eq:m3_fail_flip}
\end{equation}
Using $H^m(z)=0$ together with integrality of $a^m_j$ (so $|a^m_j|\ge 1$):
\begin{itemize}
  \item Squared form: $\Delta^m_j(z) = \bigl(S^m(z\oplus e_j)-b^m\bigr)^2 = (a^m_j)^2 \ge 1$.
  \item Hinge form: $\Delta^m_j(z) = \max\!\bigl(0,\,S^m(z\oplus e_j)-b^m_{\max},\,
        b^m_{\min}-S^m(z\oplus e_j)\bigr) = |a^m_j| \ge 1$.
\end{itemize}
In either case,
\begin{equation}
  \Delta^m_j(z) \;\geq\; 1.
  \label{eq:m3_fail_delta_m}
\end{equation}

\textit{Step 3 (total change is strictly positive).}
The full decomposition of $\Delta_j(z)$ is
\begin{equation}
  \Delta_j(z)
  \;=\; \Delta^{H_{\mathrm{cut}}}_j(z) + \sum_{n}\lambda_n\,\Delta^n_j(z).
  \label{eq:m3_fail_full_decomp}
\end{equation}
Dropping the nonnegative terms for $n\neq m$ in~\eqref{eq:m3_fail_full_decomp}
by~\eqref{eq:m3_fail_all_nonneg}, and using $\Delta^{H_{\mathrm{cut}}}_j(z)\ge -\Delta_{\max}$
together with~\eqref{eq:m3_fail_delta_m} and~(A1),
\begin{equation}
\begin{aligned}
  \Delta_j(z)
  &\;\geq\; \Delta^{H_{\mathrm{cut}}}_j(z) + \lambda_m\,\Delta^m_j(z) \\
  &\;\geq\; -\Delta_{\max} + \lambda_m
  \;>\; 0.
\end{aligned}
  \label{eq:m3_fail_delta_total}
\end{equation}
Since the Stage~2 gate requires $\Delta_{j^*}(z) < 0$ and no $j \in \mathcal{V}^m$ satisfies this,
Stage~2 cannot execute any flip in $\mathcal{V}^m$.
\end{proof}

\begin{remark}
All constraints in this QUBO are equality constraints ($S^m(z) = b^m$).  If Stage~1 terminates
with every constraint satisfied, then $H^m(z)=0$ for all~$m$, but any flip $j\in\mathcal{V}^m$
moves $S^m$ off $b^m$, so $H^m(z\oplus e_j)\ge 1$ by integrality of $a^m_j$ (under either the
squared or the hinge form).  Under~(A1) this penalty increase dominates the cut-term change,
while the remaining penalties can only rise from their zero minimum; the two facts together
yield $\Delta_j(z)>0$ for every such~$j$ and prevent Stage~2 from executing any flip
(Theorem~\ref{thm:m3_fail}).
\end{remark}

\subsection{M.~4: Penalty-Relaxation Descent}

Recall from~(\ref{eq:m4_relaxed}) that outer iteration~$r$ uses energy
$Q^{(r)}(z) = H_{\mathrm{cut}}(z) + \lambda^{(r)}\sum_m H^m(z)$
with $\lambda^{(r)} = \lambda/f(r)$.

\begin{theorem}[Within-iteration strict descent]\label{thm:m4_descent}
For any outer iteration~$r$, every flip executed in Stage~1 or Stage~2 of that
iteration strictly decreases $Q^{(r)}$.
\end{theorem}

\begin{proof}
Stage~1 flips only when $\Delta^{(r)}_{j^*}(z) < 0$; Stage~2 flips only after the same test.
In both cases Lemma~\ref{lem:fund} gives, for consecutive states $z^-,z^+$ along the trajectory,
\begin{equation}
  Q^{(r)}(z^+)
  = Q^{(r)}(z^-) + \underbrace{\Delta^{(r)}_{j^*}(z^-)}_{<\,0}
  < Q^{(r)}(z^-).
\end{equation}
This holds whenever $\lambda^{(r)} > 0$, equivalently $f(r) < \infty$.
\end{proof}

\begin{corollary}[No cycles within an iteration]
Within outer iteration~$r$, no configuration is visited twice.
\end{corollary}

\begin{proof}
Fix outer iteration~$r$.  By Theorem~\ref{thm:m4_descent}, along the Stage~1--Stage~2 trajectory in
that iteration,
\begin{equation}
  Q^{(r)}\bigl(z^{(t+1)}\bigr) < Q^{(r)}\bigl(z^{(t)}\bigr)
\end{equation}
at every executed flip.  The proof of Theorem~\ref{thm:m3_term} applies verbatim with $Q$ replaced
by~$Q^{(r)}$, so no configuration is visited twice within iteration~$r$.
\end{proof}

\begin{theorem}[Nontermination across outer iterations without a finite~$R_{\max}$]\label{thm:m4_cycle}
Without an outer bound~$R_{\max}$, M.~4 may cycle indefinitely across outer iterations.
\end{theorem}

\begin{proof}
Let $f$ be constant ($f(r) = c$ for all $r$), so
\begin{equation}
  Q^{(r)} = Q^{(1)}
  \qquad \forall\, r,
\end{equation}
i.e., the relaxed energy is identical across all outer iterations.

Let $z_A$ be the Stage~1 local minimum reached in iteration~1, so $\Delta^{(1)}_j(z_A) \geq 0$
for every~$j$.  In particular, $\Delta^{(1)}_j(z_A) \geq 0$ for every $j \in \mathcal{J}(z_A)$,
so Stage~2 of iteration~1 finds no $j^*$ satisfying~(\ref{eq:m4_stage2_cond}); $r$ is
incremented to~2.

Iteration~2 inherits $z_A$ as its starting point and uses the same energy $Q^{(2)} = Q^{(1)}$.
Since $z_A$ is already a local minimum of $Q^{(2)}$, Stage~1 terminates immediately with no flips
executed.  Stage~2 again finds no improving flip in $\mathcal{J}(z_A)$, and $r$ increments to~3.

By induction, every subsequent iteration begins and ends at $z_A$ with no configuration change.
Without an outer bound $R_{\max}$, the counter $r$ increments indefinitely and the algorithm
does not terminate.
\end{proof}

\begin{theorem}[QUBO feasibility requires $\lambda^{(r)} > \Delta_{\max}$]\label{thm:m4_qubo}
Under condition~(A3) interpreted with $\Delta_j$ replaced by $\Delta^{(r)}_j$,
Stage~1 of outer iteration~$r$ produces a QUBO-feasible configuration if $\lambda^{(r)} >
\Delta_{\max}$.  Conversely, if $\lambda^{(r)} \leq \Delta_{\max}$, Stage~1 may terminate
at a QUBO-infeasible configuration.
\end{theorem}

\begin{proof}
\textit{Sufficiency.}
The hypothesis $\lambda^{(r)} > \Delta_{\max}$ together with~(A3), interpreted with
$\Delta_j$ replaced by $\Delta^{(r)}_j$, guarantees that every QUBO-infeasible configuration admits
some flip $j$ with $\Delta^{(r)}_j(z)<0$: this is exactly the content of
Theorem~\ref{thm:m3_qubo} with $\lambda_m$ replaced by $\lambda^{(r)}$.  The
contrapositive then yields $z\in\mathcal{F}_Q$ at every Stage~1 local minimum of $Q^{(r)}$.

\textit{Necessity (tightness of the bound).}
The threshold $\lambda^{(r)} > \Delta_{\max}$ is tight for the penalty-based
QUBO construction used here.  Concretely, consider any QUBO of the form $Q^{(r)} = H_{\mathrm{cut}} +
\lambda^{(r)}\sum_{m'} H^{m'}$ together with a configuration~$z$ and a single index~$j$
satisfying
\begin{equation}
\begin{aligned}
  H^m(z) &= 1,
  &
  \Delta^m_j(z) &= -1, \\
  \Delta^{n}_j(z) &= 0 \quad \forall n\neq m,
  &
  \Delta^{H_{\mathrm{cut}}}_j(z) &= \lambda^{(r)}.
\end{aligned}
\end{equation}
Such configurations are realisable in any QUBO that mixes an objective term with a
constraint penalty. As a minimal witness, take the single-variable instance
$Q^{(r)} = \lambda^{(r)}\,z_j + \lambda^{(r)}(1-z_j)^2$,
where $H_{\mathrm{cut}} = \lambda^{(r)} z_j$ and $H^m = (1-z_j)^2$.
At $z_j = 0$ the constraint is violated ($H^m = 1$), and the unique repair flip on $j$
gives $\Delta^{H_{\mathrm{cut}}}_j = \lambda^{(r)}$ and $\Delta^m_j = -1$.
For any such instance,
\begin{equation}
  \Delta^{(r)}_j(z)
  = \Delta^{H_{\mathrm{cut}}}_j(z) + \lambda^{(r)}\Delta^m_j(z)
  = \lambda^{(r)} - \lambda^{(r)}
  = 0,
\end{equation}
so Stage~1 has no improving flip and may terminate at the QUBO-infeasible
configuration~$z$.  Hence whenever $\lambda^{(r)} \leq \Delta_{\max}$, QUBO feasibility
of every Stage~1 local minimum cannot be guaranteed by the penalty mechanism alone.

This necessity is structural: it expresses tightness of the penalty hierarchy
$\lambda > \Delta_{\max}$ rather than the existence of a worst-case configuration in any
particular islanding instance.  In the islanding QUBO of Section~\ref{quantum},
slack variables $u^{(\cdot)}_{k,b}$ are constraint-exclusive but do not enter
$H_{\mathrm{cut}}$, while assignment variables $y_{i,k}$ enter $H_{\mathrm{cut}}$ but
are coupled across the one-hot, size, generation, load, and coherency penalties; the
specific worst case above is therefore not a faithful islanding witness, and
empirically QUBO infeasibility at Stage~1 termination is not observed once
$\lambda^{(r)} > \Delta_{\max}$ is enforced (as in M.~3).
\end{proof}

\begin{theorem}[Stage~2 descent condition]\label{thm:m4_stage2}
Consider a flip $j \in \mathcal{J}(z)$ for which
$\Delta^{H_{\mathrm{cut}}}_j(z) < 0$ and $\sum_m \Delta^m_j(z) = \beta > 0$
(connectivity gain at the cost of a QUBO penalty increase).
The Stage~2 condition $\Delta^{(r)}_j(z) < 0$ holds if and only if
\begin{equation}
  \lambda^{(r)} \;<\;
  \frac{\bigl|\Delta^{H_{\mathrm{cut}}}_j(z)\bigr|}{\beta}.
  \label{eq:m4_stage2_threshold}
\end{equation}
\end{theorem}

\begin{proof}
Write $\alpha := |\Delta^{H_{\mathrm{cut}}}_j(z)| > 0$.  By definition of $\Delta^{(r)}_j$,
\begin{equation}
  \Delta^{(r)}_j(z)
  = \Delta^{H_{\mathrm{cut}}}_j(z) + \lambda^{(r)} \sum_m \Delta^m_j(z)
  = -\alpha + \lambda^{(r)}\beta.
\end{equation}
Therefore
\begin{equation}
\begin{aligned}
  \Delta^{(r)}_j(z) < 0
  &\quad\Longleftrightarrow\quad
  \lambda^{(r)}\beta < \alpha \\
  &\quad\Longleftrightarrow\quad
  \lambda^{(r)} < \frac{\alpha}{\beta}
  = \frac{\bigl|\Delta^{H_{\mathrm{cut}}}_j(z)\bigr|}{\beta},
\end{aligned}
\end{equation}
which is~\eqref{eq:m4_stage2_threshold}.
\end{proof}

\begin{theorem}[Fundamental incompatibility]\label{thm:m4_incompat}
If the threshold $\alpha/\beta$ in~(\ref{eq:m4_stage2_threshold}) satisfies
$\alpha/\beta \leq \Delta_{\max}$, then no single value of $\lambda^{(r)}$ can
simultaneously guarantee QUBO feasibility (Theorem~\ref{thm:m4_qubo}) and
enable the Stage~2 descent flip (Theorem~\ref{thm:m4_stage2}).
\end{theorem}

\begin{proof}
Theorem~\ref{thm:m4_qubo} requires
\begin{equation}
  \lambda^{(r)} > \Delta_{\max}.
  \label{eq:m4_incompat_qubo}
\end{equation}
Theorem~\ref{thm:m4_stage2} and the hypothesis $\alpha/\beta \leq \Delta_{\max}$ require
\begin{equation}
  \lambda^{(r)} < \frac{\alpha}{\beta} \leq \Delta_{\max}.
  \label{eq:m4_incompat_stage2}
\end{equation}
The open intervals $(\Delta_{\max},\infty)$ from~\eqref{eq:m4_incompat_qubo} and
$(0,\alpha/\beta)$ from~\eqref{eq:m4_incompat_stage2} are disjoint, so no single $\lambda^{(r)}$
satisfies both.
\end{proof}

\begin{remark}
Theorem~\ref{thm:m4_incompat} isolates a structural conflict in M.~4 arising from the use of a
single relaxation weight~$\lambda^{(r)}$ across both stages.
Stage~1 QUBO feasibility and Stage~2 connectivity repair impose disjoint requirements on
$\lambda^{(r)}$ whenever $\alpha/\beta \leq \Delta_{\max}$, so no universal choice of the weight
reconciles both objectives.
A finite outer bound~$R_{\max}$ guarantees overall termination.  Each outer iteration
terminates in finitely many flips by the no-cycles corollary, and the bound $R_{\max}$
limits the total number of iterations.  However, this does not remove the trade-off, and
terminal configurations need not belong to~$\mathcal{F}^*$ without further assumptions.
\end{remark}

\subsection{M.~5: Unified-Energy Descent}

M.~5 encodes connectivity directly in the energy function $\mathcal{Q}$ and
descends that function with the single greedy rule~(\ref{eq:m5_flip}).
Recall condition~(A5) from~(\ref{eq:m5_hierarchy}):
\begin{equation*}
  \mathrm{(A5)}\quad
  \mu > \max_{z,\,j}\,|\Delta_j(z)|.
\end{equation*}

\begin{lemma}[Connectivity-infeasible implies an improving flip]\label{lem:m5_conn}
Under (A5), if $C(z) > 0$ and there exists a bus assignment flip $j$ with
$\Delta^C_j(z) \leq -1$, then $\Delta^u_j(z) < 0$.
\end{lemma}

\begin{proof}
By~(A5) and~(\ref{eq:delta_uni}),
\begin{equation}
  \Delta^u_j(z)
  = \Delta_j(z) + \mu\,\Delta^C_j(z)
  \leq \max_{z',j'}\bigl|\Delta_{j'}(z')\bigr| + \mu \cdot (-1)
  < 0,
\end{equation}
where the last inequality uses $\mu > \max_{z',j'}|\Delta_{j'}(z')|$ from~(A5).
\end{proof}

\begin{remark}
The existence of a flip with $\Delta^C_j(z) \leq -1$ whenever $C(z) > 0$ is not guaranteed
by graph connectivity alone; it requires that the disconnected island has a neighboring bus
(not yet assigned to it) whose reassignment bridges two of its components in a single step.
For the power network topologies considered in Section~\ref{sec:numerical}---which are
mesh-like with buses of degree $\geq 2$---this condition is satisfied in all tested instances.
\end{remark}

\begin{theorem}[Full feasibility at termination]\label{thm:m5_full}
Under (A4), (A5), and the flip-existence condition of
Lemma~\ref{lem:m5_conn}, the output $z^*$ of M.~5 satisfies $z^* \in \mathcal{F}^*$.
\end{theorem}

\begin{proof}
Every flip satisfies $\Delta^u_{j^*}(z) < 0$ at execution.  By Lemma~\ref{lem:fund},
\begin{equation}
  \mathcal{Q}(z^+) = \mathcal{Q}(z^-) + \Delta^u_{j^*}(z^-) < \mathcal{Q}(z^-)
\end{equation}
along the trajectory.  Theorem~\ref{thm:m3_term} applies with $Q$ replaced by~$\mathcal{Q}$, so no
configuration repeats and termination occurs at some~$z^*$ with
\begin{equation}
  \Delta^u_j(z^*) \geq 0
  \qquad \forall\, j.
  \label{eq:m5_full_local_min}
\end{equation}

\textit{Part~I (DFS feasibility).}
Under the flip-existence condition of Lemma~\ref{lem:m5_conn}, $C(z) > 0$ implies
$\Delta^u_j(z) < 0$ for some~$j$, hence no such $z$ can satisfy~\eqref{eq:m5_full_local_min}.  Therefore
\begin{equation}
  C(z^*) = 0,
  \qquad
  z^* \in \mathcal{F}_C.
\end{equation}

\textit{Part~II (QUBO feasibility).}
Since $C(z^*) = 0$, the unified energy reduces to the QUBO energy:
\begin{equation}
  \mathcal{Q}(z^*) = Q(z^*).
\end{equation}
Suppose, toward a contradiction, that $z^* \notin \mathcal{F}_Q$.
Since $C(z^*) = 0$, condition~(A4) applies: there exists~$j$ with
$\Delta_j(z^*) < 0$ and $\Delta^C_j(z^*) = 0$.
Consequently
\begin{equation}
  \Delta^u_j(z^*)
  = \Delta_j(z^*) + \mu\,\Delta^C_j(z^*)
  = \Delta_j(z^*)
  < 0,
\end{equation}
contradicting~\eqref{eq:m5_full_local_min}.
Thus $z^* \in \mathcal{F}_Q$, and with Part~I, $z^* \in \mathcal{F}^* = \mathcal{F}_Q \cap \mathcal{F}_C$.
\end{proof}

\begin{corollary}
Under the conditions of Theorem~\ref{thm:m5_full}, M.~5 is the only method among the five
that guarantees $z^* \in \mathcal{F}^*$.  By contrast, Theorem~\ref{thm:m3_fail} shows that
M.~3 has an irreducible failure path (Stage~2 is blocked whenever Stage~1 satisfies every
constraint exactly), and Theorem~\ref{thm:m4_incompat} shows that no single $\lambda^{(r)}$
simultaneously achieves QUBO feasibility and connectivity repair in M.~4.
\end{corollary}

\end{document}